\documentclass[a4paper,onecolumn,superscriptaddress,11pt,accepted=2019-06-01]{quantumarticle}
\pdfoutput=1
\usepackage[utf8]{inputenc}
\usepackage[english]{babel}
\usepackage[T1]{fontenc}
\usepackage{amsmath}
\usepackage{amssymb}
\usepackage{amsthm}
\usepackage{revsymb}
\usepackage{graphicx,color}
\usepackage{bm}
\usepackage{cite}
\usepackage{enumitem}   
\usepackage{hyperref}

\usepackage{tikz}
\usepackage{lipsum}

\newtheorem{thm}{Theorem}
\newtheorem{prop}{Proposition}
\newtheorem{lemma}{Lemma}
\newtheorem{corol}{Corollary}

\theoremstyle{remark}
\newtheorem{remark}{Remark} 

\theoremstyle{definition}
\newtheorem{example}{Example}

\numberwithin{equation}{section}

\newcommand{\tr}{\mathop{\mathrm{tr}}\nolimits}
\renewcommand{\Re}{\mathop{\mathrm{Re}}\nolimits}

\newcommand{\rank}{\mathop{\mathrm{rank}}\nolimits}
\renewcommand{\d}{\mathrm{d}}
\newcommand{\rmi}{\mathrm{i}}

\newcommand{\rme}{\mathrm{e}}
\newcommand{\ket}[1]{|{#1}\rangle}
\newcommand{\bra}[1]{\langle{#1}|}

\renewcommand{\bullet}{\,\,\begin{picture}(-1,1)(-1,-3)\circle*{2}\end{picture}\ \,\,}

\definecolor{dgreen}{rgb}{0,0.5,0}

\definecolor{delete}{cmyk}{0.5,0,0,0}

\begin{document}
\title[Generalized Adiabatic Theorem and Strong-Coupling Limits]{Generalized Adiabatic Theorem and Strong-Coupling Limits}
\date{\today}
\author{Daniel Burgarth}
\orcid{0000-0003-4063-1264}
\affiliation{Center for Engineered Quantum Systems, Dept.\ of Physics \& Astronomy, Macquarie University, 2109 NSW, Australia}
\author{Paolo Facchi}
\orcid{0000-0001-9152-6515}
\affiliation{Dipartimento di Fisica and MECENAS, Universit\`a di Bari, I-70126 Bari, Italy}
\affiliation{INFN, Sezione di Bari, I-70126 Bari, Italy}
\author{Hiromichi Nakazato}
\orcid{0000-0002-5257-7309}
\affiliation{Department of Physics, Waseda University, Tokyo 169-8555, Japan}
\author{Saverio Pascazio}
\orcid{0000-0002-7214-5685}
\affiliation{Dipartimento di Fisica and MECENAS, Universit\`a di Bari, I-70126 Bari, Italy}
\affiliation{INFN, Sezione di Bari, I-70126 Bari, Italy}
\affiliation{Istituto Nazionale di Ottica (INO-CNR), I-50125 Firenze, Italy}
\author{Kazuya Yuasa}
\orcid{0000-0001-5314-2780}
\affiliation{Department of Physics, Waseda University, Tokyo 169-8555, Japan}
\maketitle

\begin{abstract}
We generalize Kato's adiabatic theorem to nonunitary dynamics with an isospectral generator. This enables us to unify two strong-coupling limits: one driven by fast oscillations under a Hamiltonian, and the other driven by strong damping under a Lindbladian. 
We discuss the case where both mechanisms are present and provide nonperturbative error bounds. 
We also analyze the links with the quantum Zeno effect and dynamics.
\end{abstract}


\section{Introduction}

Strong coupling and adiabatic decoupling are effective tools towards quantum control strategies, whose primary objective is to preserve the coherence of quantum mechanical systems~\cite{lidarbook}.
These control procedures are manifestations of the quantum Zeno effect (QZE)~\cite{ref:QZEMisraSudarshan} and its extension, known as the quantum Zeno dynamics (QZD)~\cite{ref:QZS}. They all represent robust tools for quantum control and lead to essentially the same physics. For instance, instead of repeated projective measurements, one might consider the application of strong fields or strong damping to induce the QZE and hinder part of the evolution\@.

Yet, from a mathematical perspective the proof techniques used are often very different, and it is sometimes hard to understand the relationships between the various manifestations. In particular, it is not possible to describe situations in which several mechanisms are present simultaneously. This point is however important, because physical systems are often not ``clean'' enough to focus on only one type of quantum dynamics. Our principal motivation is therefore to develop a common mathematical framework unifying such physics.

To this end, we shall focus on the ``continuous'' formulation of the QZE, in which the system to be studied is strongly coupled to an external field  or to a reservoir that induces (fast) dissipation. 
The case of strong fields was proven by making use of the adiabatic theorem for unitary dynamics~\cite{ref:QZS}, and, alternatively, with a unitary perturbation theory~\cite{ref:PaoloAdiabatic}. Averages over fastly oscillating terms are also commonly found in other areas of mathematical physics, in particular in deriving effective master equations for open quantum systems in the weak-coupling limit \cite{Davies}.
The case of strong damping was analyzed using dissipative perturbation theory~\cite{ref:ZanardiDFS-PRL2014,ref:ZanardiDFS-PRA2015}. While the proofs for the two cases are clearly related, it is difficult to see the precise connection. Here, we shall develop a suitable extension of Kato's adiabatic theorem~\cite{ref:KatoAdiabatic}, by relaxing the condition of the unitarity of the evolution. This extension allows us to unify and simplify the proof of the continuous QZE, and to describe also the situation in which both field and dissipation are present at the same time. We can furthermore provide nonperturbative error bounds for both cases.
Since the physical mechanisms at the basis of QZD have been proven in a number of recent experiments~\cite{Itano,ref:QZEExp-Ketterle,ref:QZEExp-Schafer:2014aa,QZDSignoles,QZDBretheau,QZDBarontini,ref:ZenoDiamond-Kalb:2016aa}, our results might be relevant also in the light of possible future experimental implementations.

\section{Different Manifestations of the QZD}
\label{diffmanif}
Let us briefly recapitulate the main features of the different manifestations of the QZD~\cite{ref:ControlDecoZeno,ref:PaoloSaverio-QZEreview-JPA} (see also Refs.~\cite{ref:QZE-PetroskyTasakiPrigogine,ref:SchulmanJumpTimePRA,ref:QZS,ref:BBZeno,ref:KoshinoShimizuQZE,ref:QZEExp-Ketterle,ref:QZEExp-Schafer:2014aa,ref:ZanardiDFS-PRL2014,ref:ZanardiDFS-PRA2015}).
Here, we in particular focus on the continuous strategies.

The standard way to induce the QZD is to frequently perform projective measurements~\cite{ref:QZEMisraSudarshan,ref:InverseZeno,ref:artzeno}.
Consider a finite-dimensional quantum system with a Hamiltonian $H$.
The unitary evolution of its density matrix $\rho$ is given by $\rme^{-\rmi t\mathcal{H}}(\rho) = \rme^{-\rmi tH} \rho \rme^{\rmi tH}$, with $\mathcal{H}=[H,\bullet]$. During the time interval $t$, one performs $n$ projective measurements at regular time intervals $t/n$.
A measurement is represented by a set of orthogonal projection operators $\{P_k\}$ acting on the Hilbert space, satisfying $P_kP_\ell=P_k\delta_{k\ell}$ and $\sum_kP_k=I$.
We consider nonselective measurements~\cite{schwinger,peres}, described by a projection $\mathcal{P}$ acting on a density operator $\rho$ as
\begin{equation}
  \mathcal{P}(\rho)=\sum_kP_k\rho P_k.
  \label{eqn:Pdef}
\end{equation}
In the limit of infinitely frequent measurements (\textit{Zeno limit}), the evolution of the system is described by
\begin{equation}
  \left(\mathcal{P}\rme^{-\rmi\frac{t}{n}\mathcal{H}}\right)^n
  \to \rme^{-\rmi t\mathcal{P}\mathcal{H}\mathcal{P}}\mathcal{P}
  =\rme^{-\rmi t\mathcal{H}_Z}\mathcal{P}
  \quad \text{as} \quad n\to +\infty,
  \label{eqn:ZenoLimit}
\end{equation}
where
$\mathcal{H}_Z=[H_Z,\bullet]$ and
\begin{equation}
	H_Z=\sum_kP_kHP_k.
	\label{eqn:ZenoHamiltonian}
\end{equation}
Here and in the following the convergence is in a suitable norm (in finite dimensions all norms are equivalent). In the Zeno limit, all transitions among the subspaces specified by the projection operators $\{P_k\}$ are suppressed.
The system unitarily evolves within the (\textit{Zeno}) subspaces defined by the projection operators $\{P_k\}$ under the action of the projected (\textit{Zeno}) Hamiltonian $H_Z$.
This is the QZD~\cite{ref:QZS}\@.

A QZD can also be induced by continuous strategies, that do not make use of the pulsed controls described above. Below, we briefly review the two main possibilities.

\paragraph{(i) Strong continuous coupling/fast oscillations:}
QZD can be induced by a \textit{strong continuous coupling/fast oscillations}~\cite{ref:SchulmanJumpTimePRA,ref:QZS}.
We add a Hamiltonian $\gamma K$ with a large coupling constant $\gamma$, so that the total Hamiltonian reads $\gamma K+H$.
In the strong-coupling limit, 
we get~\cite{ref:QZS}
\begin{equation}
  \rme^{-\rmi t(\gamma \mathcal{K}+\mathcal{H})}
  \sim \rme^{-\rmi t\gamma \mathcal{K}}\rme^{-\rmi t\mathcal{H}_Z}
  \quad \text{as} \quad \gamma \to +\infty,
  \label{eqn:StrongCoupling}
\end{equation}
where 
$\mathcal{K}=[K,\bullet]$ and $H_Z$ is defined again by Eq.~(\ref{eqn:ZenoHamiltonian}), but the projections $\{P_k\}$ are the spectral projections of the Hamiltonian $K=\sum_k\varepsilon_kP_k$.\footnote{In the following $A(\gamma)\sim B(\gamma)$, as $\gamma\to+\infty$, will mean that $A(\gamma)=B(\gamma) + o(B(\gamma))$, where $o(B(\gamma))$ is an operator such that $\| o(B(\gamma)) \| /  \|B(\gamma)\| \to 0$ as $\gamma \to+\infty$. On the other hand, $A(\gamma) = \mathcal{O}( B(\gamma))$ will mean that $\| A(\gamma) \| /  \|B(\gamma)\|$ is bounded as $\gamma \to+\infty$.}

\paragraph{(ii) Strong damping:} 
QZD can also be induced via \textit{strong damping}~\cite{ref:NoiseInducedZeno,ref:ZanardiDFS-PRL2014,ref:ZanardiDFS-PRA2015,Madalin,ref:VictorPRX,ref:ZanardiDFS-PRA2017}.
Suppose that the system, governed by a Hamiltonian $H$, is exposed to a strong Markovian damping process $\rme^{t\mathcal{D}}$ with $t> 0$ which steers the system into invariant (decoherence-free) subspaces, as
\begin{equation}
	\rme^{t\mathcal{D}}
	\to \mathcal{P}_\varphi
	\quad \text{as} \quad t\to +\infty,
\end{equation}
where $\mathcal{P}_\varphi$ 
is the completely positive and trace-preserving (CPTP) projection onto the invariant subspaces.\footnote{The subscript $\varphi$ is put on $\mathcal{P}_\varphi$ to stress that it is the projection onto the ``peripheral'' spectrum of $\rme^{t\mathcal{D}}$, i.e.~the eigenvalues of modulus 1 on the boundary of the unit disk on the complex plane, within which the spectrum of any CPTP maps is confined. See Proposition~\ref{eqn:GKLSgenerator} in Appendix \ref{qsemi}.}
Then, in the strong-damping limit, the system evolves as~\cite{ref:ZanardiDFS-PRL2014,ref:ZanardiDFS-PRA2015}
\begin{equation}
	\rme^{t(\gamma\mathcal{D}-\rmi\mathcal{H})}
	\to \rme^{-\rmi t\mathcal{P}_\varphi\mathcal{H}\mathcal{P}_\varphi}\mathcal{P}_\varphi=\rme^{-\rmi t\mathcal{H}_Z}\mathcal{P}_\varphi
	\quad \text{as} \quad \gamma \to +\infty,
	\label{eqn:QZEStrongDampUnitary}
\end{equation}
where $\mathcal{H}_Z$ is defined by $\mathcal{P}_\varphi \mathcal{H} \mathcal{P}_\varphi$.

\begin{figure}
\centering
\includegraphics[width=\textwidth]{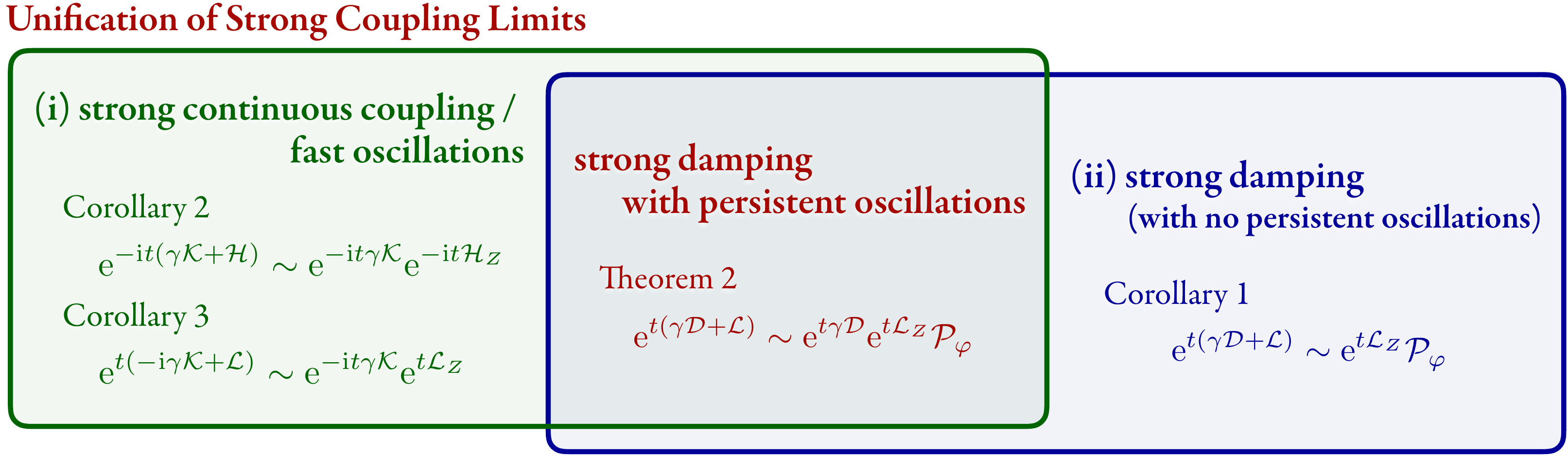}
\caption{
Summary of the main results. Strong-coupling limits, (i) via strong continuous coupling/fast oscillations and (ii) via strong damping, are unified and generalized in Theorem~\ref{thm:doublepaolo}, by making use of the generalized adiabatic theorem proven in Theorem~\ref{thm:Adiabatic}.
}
\label{fig:Summary}
\end{figure}

\bigskip
The main objective of the present article is to unify and at the same time generalize the above-mentioned continuous manifestations of the QZD in two different but inter-related formulations: (i) strong continuous coupling/fast oscillations and (ii) strong damping. With a slight abuse of language both situations are often referred to as ``strong coupling''.
We shall prove three theorems. The first one is an adiabatic theorem (Theorem~\ref{thm:Adiabatic}), that generalizes the adiabatic theorem by Kato~\cite{ref:KatoAdiabatic}, by relaxing the condition of the unitarity of the evolution.
The unified framework will then enable us to study the QZD in which the two different mechanisms (i) and (ii) are present at the same time (Theorem~\ref{thm:doublepaolo}), and generalize the QZD to nonunitary (Markovian) evolutions, with Markovian generators of the Gorini-Kossakowski-Lindblad-Sudarshan (GKLS) form~\cite{ref:DynamicalMap-Alicki,ref:GKLS-DariuszSaverio}.
These results are summarized in Fig.~\ref{fig:Summary}.
We shall also make some observations on the control of nonunitary dynamics by continuous QZD, including a no-go theorem for the cancellation of any Markovian decay (Theorem~\ref{nogolindblad}).

\section{Unification and Generalization of the Continuous QZDs}
\label{sec:MainResults}
The main result of this article is the unification and the generalization of the continuous QZDs, via (i) strong continuous coupling/fast oscillations and (ii) strong damping.
Notice that both evolutions (\ref{eqn:StrongCoupling}) and (\ref{eqn:QZEStrongDampUnitary})  
are of the form $\rme^{t(\gamma B+C)}$, and one is interested in their behavior for large values of the coupling $\gamma$, i.e.\ $\gamma \to +\infty$.
As noted in Refs.\ \cite{ref:QZS,ref:PaoloAdiabatic}, this is nothing but an adiabatic limit. Indeed, by going in the $C$ interaction picture, 
\begin{equation}
B(t) = \rme^{-tC} B  \rme^{tC},
\label{eq:cbc}
\end{equation}
the evolution $V_\gamma (t) = \rme^{-tC} \rme^{t(\gamma B + C)}$ is the solution of the equation
\begin{equation}
\dot V_\gamma (t) = \gamma B(t)  V_\gamma (t),  
\label{eq:Vg}
\end{equation}
with $V_\gamma (0) =1$. This has exactly the same form as an adiabatic evolution 
$\dot U_T (s) = T A(s) U_T (s)$
\cite{ref:KatoAdiabatic, ref:Messiah}, with the physical time and the large-coupling limit $\gamma \to +\infty$ playing the role of the scaled time $s=t/T$ and the large-time limit $T \to +\infty$, respectively.

Therefore, one is led to prove an extension to semigroups of the adiabatic theorem by Kato~\cite{ref:KatoAdiabatic}, who only considered unitary evolutions, i.e.~the case where $B$ is skew-Hermitian (Hamiltonian multiplied by the imaginary unit). 
Although this is a substantial generalization in terms of the generators allowed---including arbitrary matrices---it is rather specific in terms of their time-dependence, which is assumed to be isospectral, since 
$\mathop{\mathrm{spec}}(B(t)) = \mathop{\mathrm{spec}}(B)$. This distinguishes our work from other adiabatic theorems for open systems (e.g.~\cite{martin,schmid}).

It is worth noticing an interesting connection between the adiabatic theorem and the spectral averaging by isometries introduced by Davies~\cite{Davies,Davies80}  in his derivation of the master equation as a weak-coupling limit. 
This connection is somehow known in quantum optics, where, in the delicate derivation on physical grounds of the correct master equation for a small system interacting with a reservoir, one needs to resort to an adiabatic secular approximation~\cite[Sec.~IV.B]{CohenTannoudji}.
The spectral averaging, also known as Davies' device~\cite[Sec.~XVII.6]{vanKampen}, is in fact a rigorous way to implement the secular approximation. 

In his derivation Davies uses a fixed-point theorem, while here we find it more convenient to follow a proof technique introduced by Kato~\cite{ref:KatoAdiabatic},  that makes use of the reduced resolvent, and is more suitable to generalization to nonunitary evolutions.

We now state our results.
A few notions on quantum semigroups, which are useful for our analysis, are summarized in Appendix \ref{qsemi}, where we also review a result obtained in Ref.\ \cite{ref:ControlDecoZeno}.
\begin{thm}[Generalized adiabatic theorem]\label{thm:Adiabatic}
Let $B$ and $C$ be  operators on a finite-dimensional space. Assume that $\rme^{tB}$ is bounded for $t\geq 0$ (which is equivalent to requiring the spectrum of $B$ to be confined in the closed left half-plane $\mathbb{C}_-$, whose purely imaginary eigenvalues are semisimple), while $C$ can be an arbitrary matrix. 
Let  $\{b_k\}$ be the spectrum of $B$ and $\{P_k\}$ its spectral projections. Define
\begin{equation}
C_Z=\sum_{b_k \in \rmi \mathbb{R}} P_kCP_k,\qquad
P_\varphi=\sum_{b_k \in \rmi \mathbb{R}} P_k,
\label{eq:HZ}
\end{equation}
the Zeno projection of $C$ and the peripheral projection of $B$ (spectral projection on the purely imaginary eigenvalues of $B$), respectively.
Then, we have
\begin{equation}
\rme^{t(\gamma B+C)}=\rme^{t\gamma B}\rme^{t C_Z}+\mathcal{O}(1/\gamma) = \rme^{t C_Z } \rme^{t\gamma   B } +\mathcal{O}(1/\gamma)\quad\text{as}\quad\gamma\to +\infty ,
\label{eq:scl.0}
\end{equation}
uniformly in $t$ on compact intervals of $[0,+\infty)$. Moreover, we have
\begin{equation}
\rme^{t(\gamma B+C)}=\rme^{t\gamma B}\rme^{t C_Z}P_\varphi +\mathcal{O}(1/\gamma) = \rme^{t C_Z } \rme^{t\gamma   B }P_\varphi +\mathcal{O}(1/\gamma)\quad\text{as}\quad\gamma\to +\infty ,
\label{eq:scl}
\end{equation}
uniformly in $t$ on compact intervals of $(0,+\infty)$.
\end{thm}
\begin{proof}
The proof is given in Appendix~\ref{app:Proof}.	
\end{proof}

\begin{remark}
By gathering the corrections obtained in the proof in Appendix~\ref{app:Proof}, the distance to the adiabatic limit can be explicitly bounded e.g.~as
\begin{align}
\|\rme^{t(\gamma B + C)}-\rme^{t\gamma B}\rme^{tC_Z}P_\varphi\|
\le
\frac{1}{\gamma}\,\Biggl(&
(M+1)\sum_{b_\ell\in\rmi\mathbb{R}}\|S_\ell C P_\ell\|
\frac{M\|C\|\rme^{tM\|C\|}-\|C_Z\|\rme^{t\|C_Z\|}}{M\|C\|-\|C_Z\|}
\nonumber\\
&\qquad\qquad{}+M\|C\|\rme^{tM\|C\|}\int_0^{\infty}\d s\,\rme^{-\eta s}p(s)
\Biggr)
+\rme^{-\gamma\eta t}p(\gamma t),
\label{eqn:BoundAdiabatic}
\end{align}
for any $t\ge0$, where $M\,(\ge1)$ and $p(t)$ are some constant and some positive polynomial, respectively, such that
\begin{equation}
\|\rme^{tB}\|\le M,\qquad
\|\rme^{tB}(I-P_\varphi)\|\le\rme^{-\eta t}p(t),
\end{equation}
with a dissipative gap $\eta=\min_{b_\ell\notin\rmi\mathbb{R}}|{\Re b_\ell}|$,\footnote{When $B$ only has imaginary eigenvalues we set $\eta=\infty$.} and 
\begin{equation}
S_{\ell}
=\sum_{k\neq\ell}[(b_k -b_{\ell})I +N_k]^{-1}P_{k}
\label{eqn:Resolvent}
\end{equation}
is the reduced resolvent of $B$ at $b_\ell$\, with $N_k$ being the nilpotent of $B$ belonging to the eigenvalue $b_k$.
This bound (\ref{eqn:BoundAdiabatic}) can be further bounded as 
\begin{equation}
\|\rme^{t(\gamma B + C)}-\rme^{t\gamma B}\rme^{tC_Z}P_\varphi\|
\le
\frac{1}{\gamma}M^2
\left(
\frac{2M}{\Delta}
+\frac{1}{\eta}
\right)
\|C\|\rme^{2tM^2\|C\|}
+M
\rme^{-\gamma\eta t}
\sum_{n=0}^{D-1}\frac{1}{n!}(\gamma\eta t)^n,
\label{eqn:BoundAdiabaticSimpler}
\end{equation}
where the oscillating gap $\Delta=\min_{k\neq\ell}|b_k-b_\ell|$,\footnote{When all eigenvalues of $B$ are identical we set $\Delta=\infty$.} and $M=D\chi_\nu$, with $D$ being the dimension of $B$ and $C$, and $\chi_\nu$ being the condition number of the similarity transformation of $B$ to a Jordan form defined in terms of $\nu=\min(\eta,\Delta)$.
See Appendix~\ref{app:Bounds} for the derivation.
\end{remark}

We now apply the generalized adiabatic theorem (Theorem~\ref{thm:Adiabatic}) to generic GKLS generators, to get the unified and generalized continuous QZD\@.
In Appendix \ref{qsemi} we recall some properties of GKLS generators $\mathcal{L}$, whose exponential $\mathcal{E}_t=\rme^{t\mathcal{L}}$ is a CPTP semigroup for  $t\geq 0$, i.e.~$\mathcal{E}_t\circ \mathcal{E}_s = \mathcal{E}_{t+s}$ for $t,s\geq 0$.
The following is the central result of this article. 
\begin{thm}[QZD by strong-coupling limit]\label{thm:doublepaolo}
Let $\mathcal{L}$ and $\mathcal{D}$ be arbitrary GKLS generators. Then, we
have
\begin{equation}
\rme^{t(\gamma\mathcal{D}+\mathcal{L})}=\rme^{t \gamma \mathcal{D}}\rme^{t\mathcal{L}_Z}+\mathcal{O}(1/\gamma)\quad\text{as}\quad\gamma\to + \infty ,
\end{equation}
uniformly in $t$ on compact intervals of $[0,+\infty)$,
and
\begin{equation}
\rme^{t(\gamma\mathcal{D}+\mathcal{L})}=\rme^{t \gamma \mathcal{D}}\rme^{t\mathcal{L}_Z}\mathcal{P}_\varphi+\mathcal{O}(1/\gamma)\quad\text{as}\quad\gamma\to + \infty ,
\end{equation}
uniformly in $t$ on compact intervals of $(0,+\infty)$,
with
\begin{equation}
\mathcal{L}_Z
=\sum_{\alpha_k\in i\mathbb{R}}\mathcal{P}_k\mathcal{L}\mathcal{P}_k,\qquad
\mathcal{P}_\varphi
=\sum_{\alpha_k\in i\mathbb{R}}\mathcal{P}_k,
\label{eq:Zenosums}
\end{equation}
where $\{\alpha_k\}$ is the spectrum of $\mathcal{D}$,
$\{\mathcal{P}_k\}$ its spectral projections, and $\mathcal{P}_\varphi$ its peripheral projection.
\end{thm}
\begin{proof}
We just use Theorem~\ref{thm:Adiabatic} for $\mathcal{D}$ and $\mathcal{L}$ in place of $B$ and $C$, respectively. Proposition~\ref{eqn:GKLSgenerator} in Appendix \ref{qsemi} guarantees that the spectral assumptions on the operator $B$ are met.
\end{proof}

\begin{remark}
Since in this case $\rme^{t(\gamma\mathcal{D}+\mathcal{L})}$ and $\rme^{t\gamma\mathcal{D}}\rme^{t\mathcal{L}_Z}\mathcal{P}_\varphi$ are both CPTP and hence they are bounded by some $M\ge1$ as $\|\rme^{t(\gamma\mathcal{D}+\mathcal{L})}\|\le M$ and $\|\rme^{t\gamma\mathcal{D}}\rme^{t\mathcal{L}_Z}\mathcal{P}_\varphi\|\le M$, we can get a simpler bound than Eq.~(\ref{eqn:BoundAdiabatic}):
\begin{multline}
\|\rme^{t(\gamma\mathcal{D}+\mathcal{L})}-\rme^{t\gamma\mathcal{D}}\rme^{t\mathcal{L}_Z}\mathcal{P}_\varphi\|
\le
\frac{1}{\gamma}\,\Biggl(
M\,\biggl\|\sum_{\alpha_\ell\in\rmi\mathbb{R}}\mathcal{S}_\ell\mathcal{L}\mathcal{P}_\ell\biggr\|
\left[
2+Mt\,\Bigl(\|\mathcal{L}\|+\|\mathcal{L}_Z\|\Bigr)
\right]
\\
{}+M\|\mathcal{L}\|\int_0^{\infty}\d s\,\rme^{-\eta s}p(s)
\Biggr)
+\rme^{-\gamma\eta t}p(\gamma t),
\label{eqn:BoundCPTP}
\end{multline}
where $p(s)$ is some positive polynomial bounding $\|\rme^{t\mathcal{D}}(1-\mathcal{P}_\varphi)\|\le\rme^{-\eta t}p(t)$ with $\eta=\min_{\alpha_\ell\notin\rmi\mathbb{R}}|{\Re\alpha_\ell}|$, and $\mathcal{S}_\ell
=\sum_{k\neq\ell}(\alpha_k-\alpha_\ell+\mathcal{N}_k)^{-1}\mathcal{P}_{k}$ 
is the reduced resolvent of $\mathcal{D}$ at $\alpha_\ell$, with $\mathcal{N}_k$ being the nilpotent of $\mathcal{D}$ belonging to the eigenvalue $\alpha_k$. 
\end{remark}

Notice that this theorem unifies the two continuous QZDs, by (i) fast oscillations and (ii) strong damping.
The two mechanisms can be simultaneously present. 
The previously known results (i) and (ii) are corollaries of Theorem~\ref{thm:doublepaolo}.
Let us start with the latter one (ii): 
\begin{corol}[QZD by strong damping~\cite{ref:ZanardiDFS-PRL2014,ref:ZanardiDFS-PRA2015}]\label{prop:StrongDamp2}
Let $\mathcal{L}$ be an arbitrary GKLS generator and $\mathcal{D}$ be a GKLS generator with no persistent oscillations in the long-time limit, 
\begin{equation}
\lim_{t\to+\infty}\rme^{t\mathcal{D}}=\mathcal{P}_\varphi.
\label{eq:4.5}
\end{equation}
Then, we have
\begin{equation}
	\rme^{t(\gamma\mathcal{D}+\mathcal{L})}=\rme^{t \mathcal{L}_Z}\mathcal{P}_\varphi+\mathcal{O}(1/\gamma)\quad\text{as}\quad\gamma\to+\infty,
\end{equation}
uniformly in $t$ on compact intervals of $(0,+\infty)$,
where 
\begin{equation}
\mathcal{L}_Z=\mathcal{P}_\varphi\mathcal{L}\mathcal{P}_\varphi.
\label{eqn:LZStrongDamp}
\end{equation}
\end{corol}
\begin{proof}
Notice that in this case the assumption~\eqref{eq:4.5} implies that the only peripheral (purely imaginary) eigenvalue  is $\alpha =0$. Thus, the sums in Eq.~\eqref{eq:Zenosums} reduce to single terms, and the Zeno projection of $\mathcal{L}$ in Eq.~\eqref{eq:Zenosums} reduces to  Eq.~(\ref{eqn:LZStrongDamp}).
\end{proof}
This reproduces the previously known result on the QZD by strong damping~\cite{ref:ZanardiDFS-PRL2014,ref:ZanardiDFS-PRA2015}. For a unitary generator $-\rmi\mathcal{H}=-\rmi[H,\bullet]$ with a Hamiltonian $H$ instead of $\mathcal{L}$, this result is further reduced to Eq.~(\ref{eqn:QZEStrongDampUnitary}).

If there remain persistent oscillations in the strong-damping (long-time) limit of $\rme^{t\gamma \mathcal{D}}$, the projected generator in Eq.~(\ref {eqn:LZStrongDamp}) is further projected by the fast persistent oscillations. This situation can be treated by Theorem~\ref{thm:doublepaolo}.

On the other hand, the QZD (i), that is merely obtained by fast oscillations for unitary evolution, is reproduced by the following corollary:
\begin{corol}[QZD by fast oscillations I: projecting a Hamiltonian~\cite{ref:QZS,ref:ControlDecoZeno,ref:PaoloSaverio-QZEreview-JPA}]\label{thm:QZEFacchi}
Consider the unitary generators $-\rmi\mathcal{H}=-\rmi[H,\bullet]$ and $-\rmi\mathcal{K}=-\rmi[K,\bullet]$ with Hamiltonians $H$ and $K$. Let $K=\sum_n\varepsilon_n P_n$ be the spectral representation of $K$, with the spectrum $\{\varepsilon_n\}$ and the spectral projections $\{P_n\}$.
Then, we have
\begin{equation}
\rme^{-\rmi t(\gamma\mathcal{K}+\mathcal{H})}
=\rme^{-\rmi t\gamma \mathcal{K}}\rme^{-\rmi t\mathcal{H}_Z}+\mathcal{O}(1/\gamma)\quad\text{as}\quad \gamma\to \pm\infty,
\label{eq:sclthmunitary}
\end{equation}
uniformly in $t$ on compact intervals of $\mathbb{R}$,
with 
	\begin{equation}
		\mathcal{H}_Z
		=[\mathcal{P}_0(H),\bullet],
	\end{equation}
where
\begin{equation}
\mathcal{P}_0=\sum_nP_n\bullet P_n
\label{eqn:ProjectorForUnitary0}
\end{equation}
is the spectral projection onto the kernel of $\mathcal{K}$.
\end{corol}
\begin{proof}
We just use Theorem~\ref{thm:doublepaolo} with $\mathcal{D} = -\rmi\mathcal{K}$ and $\mathcal{L}=-\rmi\mathcal{H}$ for $t\geq0$, and with $\mathcal{D} = \rmi\mathcal{K}$ and $\mathcal{L}=\rmi\mathcal{H}$ for $t\leq0$.
Notice that in this case the spectrum of $-\rmi\mathcal{K}$ consists only of purely imaginary eigenvalues and the projection onto the peripheral spectrum of $-\rmi\mathcal{K}$ is simply the identity, $\mathcal{P}_\varphi=\sum_k\mathcal{P}_k=1$. 
By making use of Lemma~\ref{lemma:Liouvillian} in Appendix \ref{qsemi},
the spectral projections $\{\mathcal{P}_k\}$ of the GKLS generator $\mathcal{K}$ can be expressed in terms of the spectral projections $\{P_n\}$ of the Hamiltonian $K$, as 
\begin{equation}
\mathcal{P}_k = \sum_{m,n} \delta_{\omega_k, \varepsilon_m-\varepsilon_n} P_m \bullet P_n.\label{eqn:Pksum}
\end{equation}
Therefore, according to Eq.~(\ref{eq:Zenosums}), the generator $\mathcal{H}$ is projected to 
	\begin{align}
		\mathcal{H}_Z
		=\sum_k\mathcal{P}_k\mathcal{H}\mathcal{P}_k
		&=\sum_k \sum_{m,n} \sum_{m',n'} 
		\delta_{\omega_k, \varepsilon_m-\varepsilon_n} \delta_{\omega_k,\varepsilon_{m'}-\varepsilon_{n'}}
		P_m
		[H,
		P_{m'}
		\bullet
		P_{n'}
		]
		P_n
		\nonumber\displaybreak[0]\\
		&=\sum_k\sum_{m,n}
		\delta_{\omega_k, \varepsilon_m-\varepsilon_n}
		(
		P_m
		H
		P_m
		\bullet
		P_n
		-
		P_m
		\bullet
		P_n
		H
		P_n
		)
=[\mathcal{P}_0(H),\bullet].
	\end{align}
Note that, for any pairs $(m,n)$ and $(m',n')$ with a common $\omega_k=\varepsilon_m-\varepsilon_n=\varepsilon_{m'}-\varepsilon_{n'}$, the coincidence $m=m'$ implies $n=n'$, and vice versa, as $\varepsilon_m\neq \varepsilon_n$ for $m\neq n$, and in the last equality we used $\sum_k\delta_{\omega_k, \varepsilon_m-\varepsilon_n}=1$. 
\end{proof}
\begin{remark}
The Hamiltonian $H$ is projected by $\mathcal{P}_0$.
This reproduces the previously known result on the QZD by fast oscillations in Eq.~(\ref{eqn:StrongCoupling})~\cite{ref:QZS,ref:ControlDecoZeno,ref:PaoloSaverio-QZEreview-JPA}.
The unitary equivalent to Eq.~(\ref{eq:sclthmunitary}) acting on state vectors in the Hilbert space (instead of density matrices in the Liouville space) can be directly obtained as
\begin{equation}
\rme^{-\rmi t(\gamma K+H)}=\rme^{-\rmi t \gamma K}\rme^{-\rmi t H_Z}+\mathcal{O}(1/\gamma)\quad\text{as}\quad \gamma\to\pm\infty	,
\end{equation} 
with
\begin{equation}
H_Z=\sum_nP_nHP_n=\mathcal{P}_0(H),
\end{equation}
by using Theorem~\ref{thm:Adiabatic} with the Hamiltonians $\pm \rmi K$ and $\pm\rmi H$ in place of $B$ and $C$, respectively.
\end{remark}

As another corollary, Theorem~\ref{thm:doublepaolo} generalizes the QZD (i) by fast oscillations to projecting a GKLS generator, instead of a Hamiltonian. 
This is new, and is a direct consequence of Theorem~\ref{thm:doublepaolo} and of Lemma~\ref{lemma:Liouvillian}.
\begin{corol}[QZD by fast oscillations II: projecting a GKLS generator]\label{oscillationOnLindblad}	
Consider the unitary generator $-\rmi\mathcal{K}=-\rmi[K,\bullet]$ with a Hamiltonian $K$ and let $\mathcal{L}$ be an arbitrary GKLS generator. Let $\mathcal{K}=\sum_k\omega_k\mathcal{P}_k$ and $K=\sum_n \varepsilon_n P_n$ be the spectral representations of $\mathcal{K}$ and $K$, respectively, with the spectral projections $\{\mathcal{P}_k\}$ and $\{P_n\}$.
Then, we have
\begin{equation}
\rme^{t(-\rmi\gamma\mathcal{K}+\mathcal{L})}=\rme^{-\rmi t\gamma \mathcal{K}}\rme^{t \mathcal{L}_Z}+\mathcal{O}(1/\gamma)\quad\text{as}\quad\gamma\to+\infty,
\label{eq:sclthmLindblad}
\end{equation}
uniformly in $t$ on compact intervals of $[0,+\infty)$,
with 
	\begin{equation}
		\mathcal{L}_Z
		=\sum_k\mathcal{P}_k\mathcal{L}\mathcal{P}_k,
\label{eq:sclLZ}
	\end{equation}
where
\begin{equation}
\mathcal{P}_k=\sum_{m,n} \delta_{\omega_k, \varepsilon_m-\varepsilon_n} P_m\bullet P_n.
\label{eqn:ProjectorForUnitary2}
\end{equation}
\end{corol}
\begin{proof}
Use Theorem~\ref{thm:doublepaolo} with $\mathcal{D} = \pm \rmi\mathcal{K}$. Notice that in this case the spectrum of $-\rmi\mathcal{K}$ consists only of purely imaginary eigenvalues and the projection onto the peripheral spectrum of $-\rmi\mathcal{K}$ is simply the identity, $\mathcal{P}_\varphi=\sum_k\mathcal{P}_k=1$.
Then, use again Lemma~\ref{lemma:Liouvillian} in Appendix \ref{qsemi} to express the spectral projections $\{\mathcal{P}_k\}$ of $\mathcal{K}$ in terms of the spectral projections $\{P_n\}$ of $K$.
\end{proof}
\begin{remark}
The QZD by repeated nonselective projective measurements $\mathcal{P}_0=\sum_nP_n \bullet P_n$ 
[see Eq.\ (\ref{eqn:Pdef})] during the evolution $\rme^{t\mathcal{L}}$ with a GKLS generator $\mathcal{L}$ yields $(\mathcal{P}_0\rme^{\frac{t}{n}\mathcal{L}}\mathcal{P}_0)^n\to\rme^{t\mathcal{P}_0\mathcal{L}\mathcal{P}_0}\mathcal{P}_0$ as $n\to+\infty$.
This coincides with Eqs.~(\ref{eq:sclthmLindblad})--(\ref{eq:sclLZ}) by fast oscillations if the exponential of the latter is multiplied by $\mathcal{P}_0$.
Indeed, 
\begin{equation}
\rme^{t\mathcal{P}_0\mathcal{L}\mathcal{P}_0}\mathcal{P}_0=\rme^{t\sum_k\mathcal{P}_k\mathcal{L}\mathcal{P}_k}\mathcal{P}_0.
\label{eqn:notsame}
\end{equation}
However, in Eq.~(\ref{eqn:notsame}), $\mathcal{P}_0\mathcal{L}\mathcal{P}_0$ is not of GKLS form in general, while $\sum_k\mathcal{P}_k\mathcal{L}\mathcal{P}_k$ is, as follows from Eq.~(\ref{eq:sclthmLindblad}) being a CPTP semigroup. The evolution $\rme^{t\mathcal{P}_0\mathcal{L}\mathcal{P}_0}$ is governed by a licit GKLS generator only if restricted to the proper subspace by $\mathcal{P}_0$. The following example helps clarifying it.
\end{remark}
\begin{example}
Consider a qubit Hamiltonian $H=\Omega\ket{0}\bra{0}$ and a dephasing GKLS generator $\mathcal{L}=\mathcal{L}(\ket{+}\bra{+})$, where $\mathcal{L}(L)=-\frac{1}{2}\kappa(L^\dag L\bullet+\bullet L^\dag L-2L\bullet L^\dag)$ and $\ket{+}=(\ket{0}+\ket{1})/\sqrt{2}$.
In this case, the generator projected by the fast oscillations with $H$ reads $\sum_k\mathcal{P}_k\mathcal{L}\mathcal{P}_k=\frac{1}{8}\mathcal{L}(X)+\frac{1}{8}\mathcal{L}(Y)$, which is a GKLS generator, while $\mathcal{P}_0\mathcal{L}\mathcal{P}_0=\frac{1}{8}\mathcal{L}(X)+\frac{1}{8}\mathcal{L}(Y)-\frac{1}{8}\mathcal{L}(Z)$ is not a GKLS generator, where $X=\ket{0}\bra{1}+\ket{1}\bra{0}$, $Y=-\rmi(\ket{0}\bra{1}-\ket{1}\bra{0})$, and $Z=\ket{0}\bra{0}-\ket{1}\bra{1}$ are Pauli operators.
Clearly, the restricted evolution $\rme^{t\mathcal{P}_0\mathcal{L}\mathcal{P}_0}\mathcal{P}_0$ is governed by a GKLS generator. 
\end{example}

\begin{remark}
We also remark that an alternative geometric approach developed in Ref.~\cite{Rouchon} allows for a perturbative Lindbladian description of QZD\@.
\end{remark}

The QZE by fast oscillations (i), however, cannot cancel any decay due to a GKLS generator.
This no-go theorem can be proven with the help of the following proposition, generalizing an argument given in Ref.~\cite{ref:OpenSystemNotControllable}.

We first need the following decomposition.
Let $\mathcal{L}$ be an arbitrary GKLS generator. It can always be decomposed as 
\begin{equation}
\mathcal{L}=-\rmi\mathcal{H}+\mathcal{D},
\end{equation}
where  $-\rmi\mathcal{H}=-\rmi[H,\bullet]$ is its Hamiltonian part and
\begin{equation}
\mathcal{D}
=-\frac{1}{2}\sum_i
(L_i^\dag L_i\bullet+\bullet L_i^\dag L_i-2L_i\bullet L_i^\dag)
\label{eq:dissdef}
\end{equation}
is its dissipative/dephasing part, with the operators $L_i$'s assumed to be traceless, 
\begin{equation}
\tr L_i=0,\quad\forall i.
\label{eqn:Traceless}
\end{equation}
This decomposition is always possible: if $L_i$'s are not traceless, their scalar parts can always be absorbed by the Hamiltonian $H$.

\begin{prop}\label{prop:NoGoLindblad}
Let $\mathcal{L}=-\rmi\mathcal{H}+\mathcal{D}$ be an arbitrary GKLS generator, 
decomposed as above.
Consider the evolution  $\rho(t)=\mathcal{E}_t(\rho)$  of a state $\rho$, where $\mathcal{E}_t=\rme^{ t\mathcal{L}}$ and $t\geq 0$, and examine the evolution of its purity $\tr \rho^2(t)$.
Then, the largest decay rate of the purity, 
\begin{equation}
\Gamma=\sup_{\rho,t}\!\left(-\frac{\d}{\d t}\tr\rho^2(t)\right),
\end{equation}
is independent of the Hamiltonian $H$. Furthermore, 
\begin{equation}
\Gamma=0
\quad\text{iff}\quad
\mathcal{D}=0.
\end{equation}
\end{prop}
\begin{proof}
By noting the decomposition of $\mathcal{L}$ into the Hamiltonian $-\rmi\mathcal{H}=-\rmi[H,\bullet]$ and the dissipative/dephasing $\mathcal{D}$ parts, we have
\begin{equation}
\frac{\d}{\d t}\tr\rho^2(t)
=
2\tr\!\left(\rho(t)\frac{\d}{\d t}\rho(t)\right)
=2\tr\!\left[\rho(t)\,\Bigl(-\rmi[H,\rho(t)]+\mathcal{D}(\rho(t))\Bigr)\right]
=2\tr[\rho(t)\mathcal{D}(\rho(t))],
\end{equation}
which is independent of $H$. 
Moreover, if $\mathcal{D}=0$, we have $\Gamma=0$.

On the other hand, notice that 
\begin{equation}
\sup_{\rho,t}\left(-\frac{\d}{\d t}\tr\rho^2(t)\right)=\sup_{\rho}\left(\left.-\frac{\d}{\d t}\tr\rho^2(t)\right|_{t=0}\right),
\end{equation} 
because we can always choose $\rho(t)$ as an initial state.
Thus, we have
\begin{equation}
\Gamma
=\sup_{\rho}\!\left(\left.-\frac{\d}{\d t}\tr\rho^2(t)\right|_{t=0}\right)
=2 \sup_{\rho} \Bigl(-\tr[\rho \mathcal{D}(\rho)] \Bigr)
= 2 \sup_{\rho} \sum_i \tr( L_i^\dag L_i \rho^2 - L_i^\dag \rho L_i \rho).
\end{equation}
We can bound $\Gamma$ from below as
\begin{align}
\Gamma 
 = 2 \sup_{\rho} \sum_i \tr( L_i^\dag L_i \rho^2 - L_i^\dag \rho L_i \rho)
&\ge
2 \sup_{\rho=\ket{\psi}\bra{\psi}} \sum_i \tr( L_i^\dag L_i \rho^2 - L_i^\dag \rho L_i \rho)
\nonumber\displaybreak[0]\\
& = 2 \sup_{\ket{\psi} }\sum_i \Bigl(
\bra{\psi}
L_i^\dag L_i
\ket{\psi}
-
\bra{\psi}
L_i^\dag
\ket{\psi}
\bra{\psi}
L_i\ket{\psi}
\Bigr)
\ge 0,
\end{align}
where in the last inequality  we used the Cauchy-Schwarz inequality. 
If $\Gamma=0$, then the equality in the Cauchy-Schwarz inequality should hold.
It holds iff $L_i\ket{\psi}\propto\ket{\psi}$, for all $i$ and $\ket{\psi}$.
All vectors $\ket{\psi}$ are  eigenvectors of $L_i$. 
This implies that $L_i=\lambda_iI$ with some complex value $\lambda_i$ for all $i$.
But, from the assumption~(\ref{eqn:Traceless}) we have  that $\lambda_i =0$, that is  $L_i=0$ for all $i$, and hence $\mathcal{D}=0$.
\end{proof}

\begin{thm}[The QZE by fast oscillations cannot cancel any decay due to GKLS generators]\label{nogolindblad}
If $\mathcal{L}$ is a GKLS generator of nonunitary dynamics, with purity decay rate $\Gamma$, then the projected GKLS generator $\mathcal{L}_Z$ given in Eq.~(\ref{eq:sclLZ}) by fast oscillations also generates nonunitary dynamics with the same purity decay rate $\Gamma$.
\end{thm}
\begin{proof}
Since some nonunitary component is present in $\mathcal{L}$, by Proposition~\ref{prop:NoGoLindblad} the evolution $\rme^{t(-\rmi\gamma\mathcal{K}+\mathcal{L})}$ has a purity decay $\Gamma\neq 0$ independent of $\gamma$ and of the Hamiltonian~$K$. Because $\rme^{\rmi t \gamma\mathcal{K}}$ is unitary, one gets
\begin{equation}
\tr\{[(\rme^{\rmi t \gamma\mathcal{K}}\rme^{t(-\rmi\gamma\mathcal{K}+\mathcal{L})})(\rho)]^2\}
=\tr\{[\rme^{t(-\rmi\gamma\mathcal{K}+\mathcal{L})}(\rho)]^2\}, 
\end{equation}
so that $\rme^{\rmi t \gamma\mathcal{K}}\rme^{t(-\rmi\gamma\mathcal{K}+\mathcal{L})}$ has the same nonvanishing decay rate $\Gamma$ as that of $\rme^{t\mathcal{L}}$ for any $\gamma$, and 
\begin{equation}
\rme^{t \mathcal{L}_Z}=\lim_{\gamma\to+\infty}\rme^{\rmi t \gamma\mathcal{K}}\rme^{t(-\rmi\gamma\mathcal{K}+\mathcal{L})},
\end{equation}
too.
\end{proof}

\section{Example: QZD by Strong Damping and Persistent Oscillations}
\label{sec:Example}
Let us look at an example of the Zeno limit given in Theorem~\ref{thm:doublepaolo}.
We consider two GKLS generators of a three-level system:
\begin{equation}
\mathcal{L}
=-\rmi[K,\bullet]-\frac{1}{2}(L^\dag L\bullet+\bullet L^\dag L-2L\bullet L^\dag)
\end{equation}
with
\begin{gather}
K
=
\Omega_0\ket{0}\bra{0}
+\Omega_1\ket{1}\bra{1}
+\Omega_2\ket{2}\bra{2}
=
\begin{pmatrix}
	\Omega_0&&\\
	&\Omega_1&\\
	&&\Omega_2
\end{pmatrix},\\
L
=
\sqrt{\Gamma}
\,\Bigl(
\ket{1}\bra{1}
+\ket{2}\bra{2}
\Bigr)
=
\sqrt{\Gamma}
\begin{pmatrix}
	0&&\\
	&1&\\
	&&1
\end{pmatrix},
\end{gather}
and 
\begin{equation}
\mathcal{D}
=-\rmi[H,\bullet]-\frac{1}{2}(F^\dag F\bullet+\bullet F^\dag F-2F\bullet F^\dag)
\end{equation}
with
\begin{gather}
H
=
g\,\Bigl(
\ket{0}\bra{1}
+\ket{1}\bra{0}
\Bigr)
+\omega_2\ket{2}\bra{2}
=
\begin{pmatrix}
	0&g&\\
	g&0&\\
	&&\omega_2
\end{pmatrix},
\\
F
=\sqrt{\kappa}\,\ket{1}\bra{2}
=\sqrt{\kappa}
\begin{pmatrix}
	0&&\\
	&0&1\\
	&0&0
\end{pmatrix}.
\end{gather}
We look at the evolution $\rme^{t(\gamma\mathcal{D}+\mathcal{L})}$ with  large $\gamma$.
The former generator $\mathcal{L}$, describing pure dephasing between $\ket{0}$ and the other levels, is projected by the strong action of the latter generator $\mathcal{D}$.
Note that $\rme^{t \gamma \mathcal{D}}$ induces decay from $\ket{2}$ to $\ket{1}$ by $F$, and exhibits persistent oscillations between $\ket{0}$ and $\ket{1}$ at long times.
We restrict ourselves to the case $\kappa>0$ and $g>0$. 
The generator $\mathcal{L}$ is projected by the two mechanisms: the strong amplitude-damping from $\ket{2}$ to $\ket{1}$ and  the fast persistent oscillations between $\ket{0}$ and $\ket{1}$.
Notice that the unitary part and the dissipative part of $\mathcal{D}$ do not commute and the two mechanisms act nontrivially.

Let us first look at the free evolution $\rme^{t\mathcal{D}}$.
It is solved as
\begin{align}
\rme^{t\mathcal{D}}
=
\rme^{-\rmi tH}\,
\biggl[&
\Bigl(
P+\rme^{-\kappa t/2}\ket{2}\bra{2}
\Bigr)\,
\bullet
\,\Bigl(
P+\rme^{-\kappa t/2}\ket{2}\bra{2}
\Bigr)
\nonumber\displaybreak[0]\\
&{}
+
\frac{1}{2}
\,\biggl(
(
1-\rme^{-\kappa t}
)
P
-
\frac{\kappa}{\kappa^2+4g^2}
[
2g-\rme^{-\kappa t}(2g\cos2gt+\kappa\sin2gt)
]
Y
\nonumber\displaybreak[0]\\
&\qquad\qquad\qquad
{}
-
\frac{\kappa}{\kappa^2+4g^2}
[\kappa-\rme^{-\kappa t}(\kappa\cos2gt-2g\sin2gt)]
Z
\biggr)\,
\bra{2}\bullet\ket{2}
\biggr]\,\rme^{\rmi tH},
\end{align}
where $
P=\ket{0}\bra{0}+\ket{1}\bra{1}
$, $
X=\ket{0}\bra{1}+\ket{1}\bra{0}
$, $
Y=-\rmi(\ket{0}\bra{1}-\ket{1}\bra{0})
$, and $
Z=\ket{0}\bra{0}-\ket{1}\bra{1}
$.
As time $t$ goes on, it asymptotically behaves as
\begin{equation}
\rme^{t\mathcal{D}}
\sim
\rme^{-\rmi t\mathcal{H}_\infty}\mathcal{P}_\varphi\quad\text{as}\quad t\to+\infty,
\end{equation}
with the asymptotic unitary generator $-\rmi\mathcal{H}_\infty=-\rmi[gX,\bullet]$ and the projection 
\begin{equation}
\mathcal{P}_\varphi
=
P{}\bullet{}P
+
\frac{1}{2}
\left(
P
-\frac{\kappa}{\kappa^2+4g^2}
(
2gY
+
\kappa Z
)
\right)
\bra{2}\bullet\ket{2}
\end{equation}
onto the peripheral spectrum of $\mathcal{D}$.
The peripheral spectrum of $\mathcal{D}$ consists of three peripheral eigenvalues, $\alpha_0=0$ and $\alpha_\pm=\mp2ig$, with the corresponding eigenprojections given by
\begin{equation}
\begin{cases}
\medskip
\displaystyle
\mathcal{P}_0
=
\frac{1}{2}
[
P\tr(\bullet)
+
X
\tr(X\bullet)
],\\
\displaystyle
\mathcal{P}_\pm
=
\ket{\pm}
\left(
\bra{\pm}
\bullet
\ket{\mp}
-
\frac{1}{2}
\frac{\kappa}{\kappa^2+4g^2}
(
\kappa\pm2ig
)
\bra{2}\bullet\ket{2}
\right)
\bra{\mp},
\end{cases}
\end{equation}
where $\ket{\pm}=(\ket{0}\pm\ket{1})/\sqrt{2}$ are the eigenstates of $X$ belonging to the eigenvalues $\pm1$, respectively.
Then, in the Zeno limit $\gamma\to+\infty$, the generator $\mathcal{L}$ is projected to
\begin{align}
	\mathcal{L}_Z
	&=\sum_{k=0,\pm}\mathcal{P}_k\mathcal{L}\mathcal{P}_k
\nonumber\\
&
=
-\frac{1}{8}\Gamma
\left(
2X
\tr(X\bullet)
+
Y
\tr
(Y
\bullet
)
+Z
\tr
(Z
\bullet
)
-
\frac{\kappa}{\kappa^2+4g^2}
(
\kappa
Z
+
2g
Y
)
\bra{2}\bullet\ket{2}
\right).
\end{align}
See Fig.~\ref{fig:ZenoLimitZontCont} for the convergence to the Zeno dynamics.
\begin{figure}
\centering
\begin{tabular}{r@{\quad}r@{\quad}r}
&
\tiny$gt=1.0$, $\Gamma t=0.0$
&
\tiny$gt=2.0$, $\Gamma t=0.0$
\\[-2.9truemm]
\includegraphics[scale=0.523]{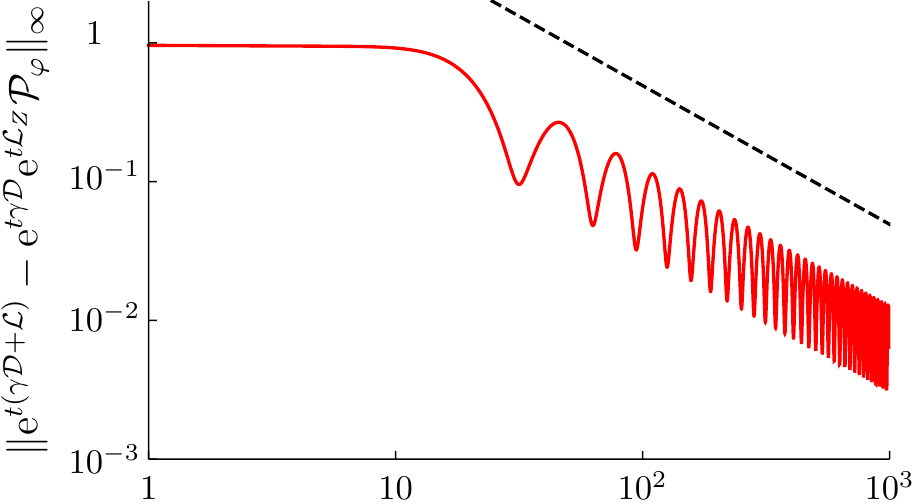}
&
\includegraphics[scale=0.523]{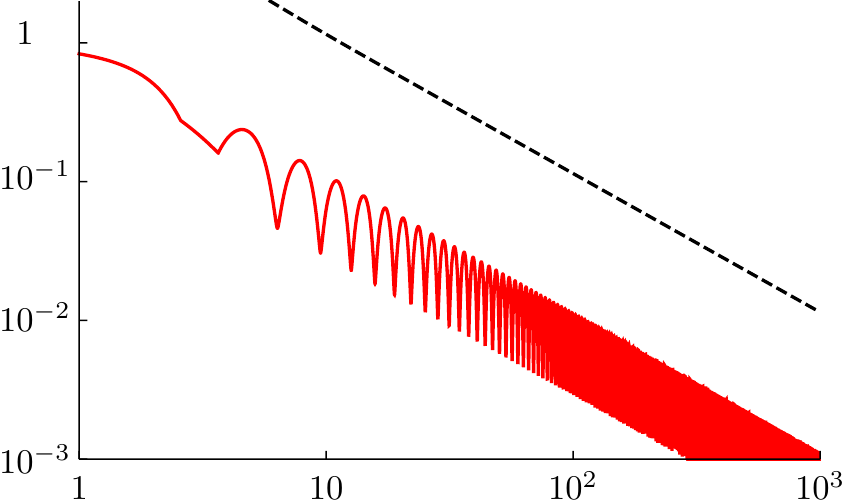}
&
\includegraphics[scale=0.523]{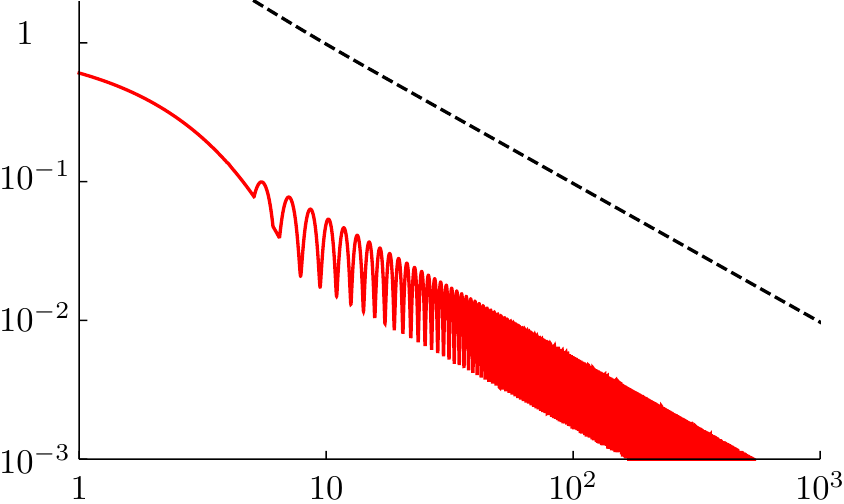}
\\[-9truemm]
\tiny$gt=0.1$, $\Gamma t=0.0$\hspace*{18truemm}
&
&
\\[4.2truemm]
&
\tiny$gt=1.0$, $\Gamma t=2.0$
&
\tiny$gt=2.0$, $\Gamma t=2.0$
\\[-2.9truemm]
\includegraphics[scale=0.522]{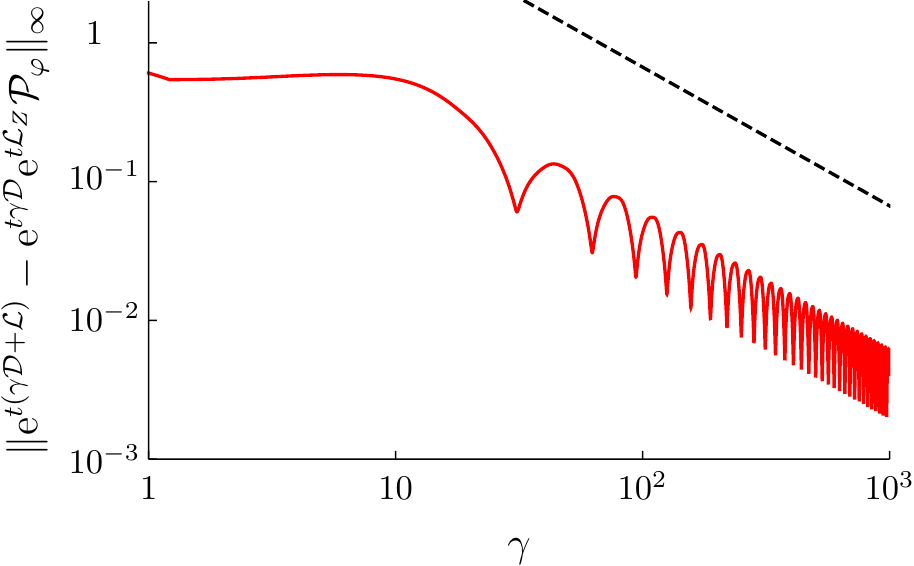}
&
\includegraphics[scale=0.522]{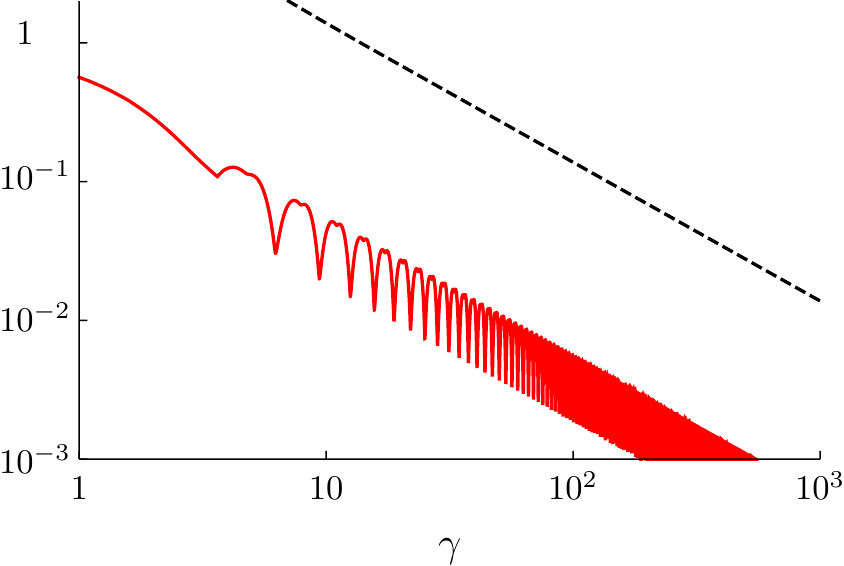}
&
\includegraphics[scale=0.522]{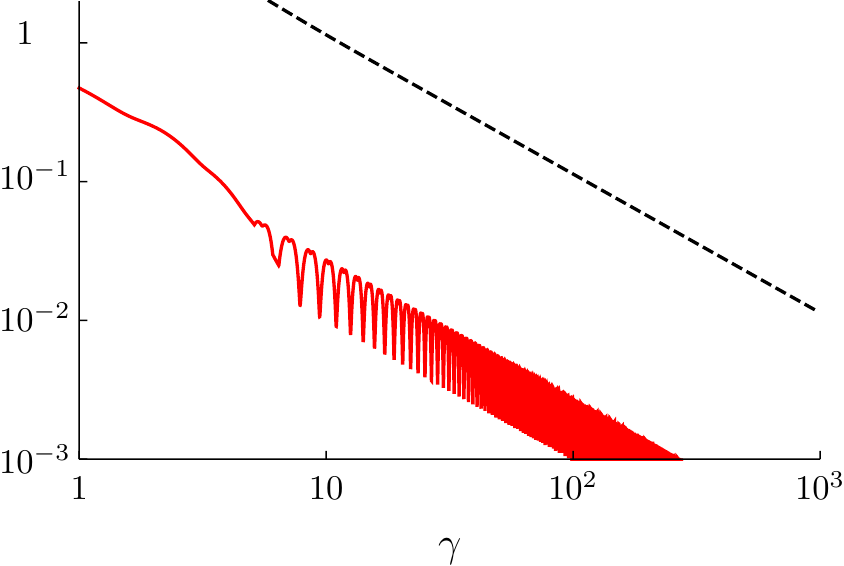}
\\[-12.7truemm]
\tiny$gt=0.1$, $\Gamma t=2.0$\hspace*{18truemm}
&
&
\\[7.9truemm]
\end{tabular}
\caption{The convergence to the Zeno dynamics via the strong damping and the persistent oscillations in the model described in Sec.~\ref{sec:Example}. 
The parameters other than $g t$ and $\Gamma t$ are set at $\Omega_0t=0.0$, $\Omega_1t=1.0$, $\Omega_2t=2.0$, $\omega_2t=1.0$, and $\kappa t=1.0$. The dashed line is the bound in Eq.~(\ref{eqn:BoundCPTP}) with $M=\sqrt{2}$, $p(s)=\sqrt{2}$, and $\eta=\kappa/2$.}
\label{fig:ZenoLimitZontCont}
\end{figure}

\section{Conclusions}
We have generalized Kato's adiabatic theorem to a particular type of nonunitary dynamics. This allowed us to treat two very different physical strong-coupling limits, strong oscillations and strong damping, on equal mathematical footing, and to generalize previous results to the situation where both effects are present simultaneously. Both limits provide various Zeno-type generators, including Zeno Hamiltonians and Zeno Lindbladians.

The adiabatic theorem allows us to provide explicit bounds on the convergence which depend on various parameters, in particular the oscillation gap and the dissipative gap of the strong generator. This demonstrates nicely both effects.

A fantastic experiment to which our generalized result is related is provided by Ref.~\cite{ref:jump}. There, we have a three-level system $\{|D\rangle,|G\rangle,|B\rangle\}$ with strong driving $\Omega_{BG} \gg \Omega_{DG}$ between $|B\rangle$ and $|G\rangle$ \emph{and} strong damping $\Gamma\gg \Omega_{DG}$  on level $|B\rangle$. What is however interesting there is that $\Gamma \gg \Omega_{BG}$. Such multi-scale strong coupling could be therefore an interesting extension of our result in future studies.

\section*{Acknowledgments}
We thank V. V. Albert and Z. K. Minev for discussions. 
DB acknowledges support by Waseda University and partial support by the EPSRC Grant No.~EP/M01634X/1\@.
This work was supported by the Top Global University Project from the Ministry of Education, Culture, Sports, Science and Technology (MEXT), Japan.
KY was supported by the Grants-in-Aid for Scientific Research (C) (No.~18K03470) and for Fostering Joint International Research (B) (No.~18KK0073) both from the Japan Society for the Promotion of Science (JSPS), and by the Waseda University Grant for Special Research Projects (No.~2018K-262).
PF and SP are  supported by INFN through the project `QUANTUM'\@. PF is supported by the Italian National Group of Mathematical Physics (GNFM-INdAM).

\appendix
\section*{Appendix}
\section{Spectral Properties of Quantum Semigroups}
\label{qsemi}
We recall here a few properties of quantum semigroups that are useful for our analysis.
Every linear operator $A$ on a finite-dimensional space can be expressed (essentially) uniquely in terms of its Jordan normal form (\textit{canonical form} or \textit{spectral representation})~\cite{ref:Katobook}:
\begin{equation}
\label{eq:Jordanform}
A = \sum_k (\lambda_k P_k + N_k),
\end{equation}
where $\{\lambda_k\}$, the \textit{spectrum} of $A$, is the set of distinct eigenvalues of $A$ ($\lambda_k\neq \lambda_\ell$ for $k\neq \ell$), $\{P_k\}$, the \textit{spectral projections} of $A$, are the corresponding eigenprojections, satisfying
\begin{equation}
\label{eq:PVM}
P_k P_\ell = \delta_{k\ell} P_k, \qquad \sum_j P_j = I,
\end{equation}
for all $k$ and $\ell$, and $\{N_k\}$ are the corresponding \textit{nilpotents} of $A$, satisfying for all $k$ and $\ell$
\begin{equation}
\label{eq:nilpot}
P_k N_\ell = N_{\ell} P_k = \delta_{k\ell} N_k, \qquad N_k^{n_k} = 0,
\end{equation}
for some integer $1\leq n_k\leq \rank P_k$.

Notice that the spectral projections, which determine a partition of the space through the \textit{resolution of identity}~\eqref{eq:PVM}, are not Hermitian in general, $P_k \neq P_k^\dagger$. An eigenvalue $\lambda_k$ of $A$ is called \textit{semisimple} or \textit{diagonalizable} if the corresponding nilpotent $N_k$ is zero (equivalently, $n_k=1$). The operator $A$ is diagonalizable if and only if all its eigenvalues are semisimple.

In the following, we state without proofs some properties 
of GKLS generators $\mathcal{L}$, whose exponential $\mathcal{E}_t=\rme^{t\mathcal{L}}$ is a CPTP semigroup for  $t\geq 0$, i.e.~$\mathcal{E}_t\circ \mathcal{E}_s = \mathcal{E}_{t+s}$ for $t,s\geq 0$.
For further details and proofs, see e.g.~Ref.~\cite{ref:Mixing-Wolf}, in particular Propositions~6.1--6.3 and Theorem~6.1 therein.
See also Refs.~\cite{ref:BaumgartnerI,ref:BaumgartnerII,ref:DiscreteBaumgartner,ref:VictorPRX,ref:VictorPhD} for the recent studies on the structure of GKLS generators, in particular on their steady states.

\begin{prop}[Spectral properties of GKLS generators]
\label{eqn:GKLSgenerator}
Let $\mathcal{L}$ be a GKLS generator on a finite-dimensional space. Then, the following properties hold:
\begin{enumerate}[label=(\roman*)]

\item The spectrum $\{\lambda_k\}$ of $\mathcal{L}$ is contained in the closed left half-plane
$\mathbb{C}_-=\{\lambda \in \mathbb{C},  \Re \lambda \leq 0\}$, and $\lambda=0$ is an eigenvalue. Moreover, all the ``peripheral'' eigenvalues, belonging to the boundary of $\mathbb{C}_-$, i.e.~on the imaginary axis $\partial\mathbb{C}_-= \rmi \mathbb{R} = \{\lambda \in \mathbb{C}, \Re \lambda = 0\}$, are semisimple.

\item The canonical form of $\mathcal{L}$  
reads
\begin{equation}
	\mathcal{L}
	=\mathcal{L}_\varphi +  \sum_{\Re \lambda_k <0} (\lambda_k \mathcal{P}_k+\mathcal{N}_k),
	\label{eqn:SpectralDecompLindblad}
\end{equation}
where 
\begin{equation}
	\mathcal{L}_\varphi
	= \sum_{\Re \lambda_k=0 }\lambda_k\mathcal{P}_k
\end{equation} 
is the ``peripheral'' part of $\mathcal{L}$, and $\{\mathcal{P}_k\}$ and $\{\mathcal{N}_k\}$ are the spectral projections and the nilpotents of $\mathcal{L}$, respectively.

\item
The projection onto the peripheral spectrum of $\mathcal{L}$,
\begin{equation}
	\mathcal{P}_\varphi=\sum_{\Re \lambda_k=0 }\mathcal{P}_k,
\end{equation}
is CPTP, and $\mathcal{L}_\varphi=\mathcal{L}\mathcal{P}_\varphi=\mathcal{P}_\varphi \mathcal{L}$.
The projection $\mathcal{P}_\varphi$ is also the  projection onto the peripheral spectrum (consisting of the eigenvalues of unit magnitude) of the CPTP map $\rme^{ t\mathcal{L}}$ for any $t>0$.
The part of $\mathcal{L}$ corresponding to the eigenvalues with negative real parts describes decay.

\item 
The peripheral map $\rme^{t\mathcal{L_\varphi}} \mathcal{P}_\varphi$ is CPTP for all $t\in \mathbb{R}$. 
\end{enumerate}
\end{prop}

In the particular case of a GKLS generator of a unitary evolution $\mathcal{L} = -\rmi \mathcal{K}= - \rmi [K,\bullet]$, where $K$ is Hermitian, all eigenvalues are imaginary and the operator reduces to its peripheral part, so that $\mathcal{P}_\varphi = 1$ and $\mathcal{L}= \mathcal{L}_\varphi$, and the generator is diagonalizable. 
In such a case the spectrum and the spectral projections of $\mathcal{K}$ can be expressed in terms of the spectrum and of the spectral projections of $K$, as shown in the following lemma.
\begin{lemma}
\label{lemma:Liouvillian}
Let $-\rmi\mathcal{K}=-\rmi[K,\bullet]$ be a GKLS generator corresponding to a Hamiltonian $K$. Let 
\begin{equation}
\mathcal{K}=\sum_k\omega_k\mathcal{P}_k, \qquad K=\sum_n \varepsilon_n P_n
\end{equation} 
be the spectral representations of $\mathcal{K}$ and $K$, respectively, with $\omega_0=0$. Then,
the spectrum of $\mathcal{K}$ is the set of the transition frequencies
$\{\omega_k\}$ determined by the energies $\{\varepsilon_n\}$, that is, for every $k$ there exists  
a pair $(m,n)$ such that
\begin{equation}
\omega_k = \varepsilon_m - \varepsilon_n,
\end{equation}
and the corresponding spectral projections are given by
\begin{equation}
\mathcal{P}_k = \sum_{m,n} \delta_{\omega_k, \varepsilon_m-\varepsilon_n} P_m \bullet P_n, \qquad
\mathcal{P}_0=\sum_nP_n \bullet P_n .
\label{eqn:ProjectorForUnitary}
\end{equation}
\end{lemma}
\begin{proof}
The second relation in Eq.~\eqref{eqn:ProjectorForUnitary} is a direct consequence of the first one and of the convention $\omega_0=0$. The spectrum of $\mathcal{K}$ and the first identity in Eq.~\eqref{eqn:ProjectorForUnitary}
are obtained by the following computation:
\begin{align}
\mathcal{K}  = [K,\bullet] 
&= \Bigl[\sum_n \varepsilon_n P_n, \bullet\Bigr]
\nonumber\\
&= \sum_{m,n} (\varepsilon_m - \varepsilon_n) P_m \bullet P_n 
\nonumber\\
& = \sum_{m,n} (\varepsilon_m - \varepsilon_n) \sum_k \delta_{\omega_k,\varepsilon_m-\varepsilon_n} P_m \bullet P_n
\nonumber\\
&=  \sum_k \omega_k \sum_{m,n}  \delta_{\omega_k,\varepsilon_m-\varepsilon_n} P_m \bullet P_n 
\nonumber\\
&= \sum_k \omega_k \mathcal{P}_k,
\end{align}
which completes the proof. This is actually a simpler proof of a result presented in Ref.~\cite{ref:ControlDecoZeno}.
\end{proof}

\section{Proof of the Generalized Adiabatic Theorem}
\label{app:Proof}
Here we prove the generalized adiabatic theorem (Theorem~\ref{thm:Adiabatic}) stated in Sec.~\ref{sec:MainResults}.
The proof is an extension to semigroups of the adiabatic theorem by Kato~\cite{ref:KatoAdiabatic}, who only considered unitary evolutions. The noteworthy difference and the novelty of this proof are the bound on the semigroup $\rme^{t(\gamma B + C)}$.
We will see that the diagonalizability (semisimpleness) of the purely imaginary spectrum plays a crucial role.

\begin{proof}[Proof of Theorem~\ref{thm:Adiabatic}]
Let 
\begin{equation}
B = \sum_{k} (b_{k}P_{k} + N_{k} )
\end{equation}
be the spectral representation of $B$, where $\{b_k\}$ is its spectrum, $\{P_k\}$ its spectral projections, and $\{N_k\}$ its nilpotents. By assumption, $\Re b_k\leq 0$ and the purely imaginary eigenvalues are semisimple (the corresponding nilpotents are zero). 

We consider separately the contributions of the peripheral eigenvalues and of the nonperipheral ones.

\paragraph{Step 1. Purely imaginary eigenvalues:}
We first focus on a specific purely imaginary eigenvalue $b_{\ell}\in \rmi \mathbb{R}$ of $B$.
It is semisimple, so that
\begin{equation}
(B -b_{\ell}I) P_{\ell}=0. 
\label{eq:eigeq}
\end{equation}
Let us define the reduced resolvent at $b_\ell$ by
\begin{equation}
S_{\ell}
=\sum_{k\neq\ell}[(b_k -b_{\ell})I +N_k]^{-1}P_{k}.
\end{equation}
Then, we have

\begin{equation}
(B-b_{\ell}I)S_{\ell}
=I-P_{\ell}. 
\label{eq:resolv}
\end{equation}
The relation~\eqref{eq:resolv} 
will play an important role later.

We are going to prove the adiabatic limit
\begin{equation}
\rme^{t(\gamma B + C)} P_\ell=\rme^{t \gamma B} \rme^{t C_Z} P_\ell+\mathcal{O}(1/\gamma)
\quad\text{as}\quad
\gamma\to +\infty,
\label{eq:K29}
\end{equation}
uniformly for $t\in[0,t_2]$ with $t_2>0$.
We start by noting
\begin{align}
\rme^{t(\gamma B + C)} -\rme^{t (\gamma B+C_Z)}  &=
-\int_0^t \d s\, \frac{\partial}{\partial s}(\rme^{(t-s) (\gamma B + C)} \rme^{s (\gamma B +C_Z)})
\nonumber\\
& =   \int_0^t \d s\, \rme^{(t-s) (\gamma B + C)} (C - C_Z) \rme^{s (\gamma B +C_Z )},
\end{align}
which multiplied by $P_\ell$ reads
\begin{equation}
\rme^{t(\gamma B + C)}P_\ell -\rme^{t (\gamma B+C_Z)} P_\ell  =  
 \int_0^t \d s\, \rme^{(t-s) (\gamma B + C)} (I-P_\ell)C P_\ell \rme^{s(\gamma b_\ell I+C_Z)}.
\label{eq:A3}
\end{equation}
We wish to prove that the integral in Eq.~\eqref{eq:A3} is of order $\mathcal{O}(1/\gamma)$---even if the integrand is $\mathcal{O}(1)$---due to its fast oscillatory behavior, e.g.~like
\begin{equation}
\int_0^t \d s\,\rme^{\rmi s \gamma}  = 
-\frac{\rmi}{\gamma}\int_0^t  \d s\,\frac{\d}{\d s} (\rme^{\rmi s \gamma})  = -\frac{\rmi}{\gamma} (\rme^{\rmi t \gamma} - 1).
\end{equation}

By using Eq.~\eqref{eq:resolv}, one can rewrite the integrand in Eq.~(\ref{eq:A3}) as
\begin{equation}
\rme^{(t-s) (\gamma B + C)} (I-P_\ell)C P_\ell \rme^{s(\gamma b_\ell I+C_Z)} 
=
\rme^{(t-s) (\gamma B + C)} (B-b_\ell I) S_\ell C P_\ell
\rme^{s(\gamma b_\ell I+C_Z)}.
\label{eq:A9}
\end{equation}
Since
\begin{equation}
 \frac{\partial}{\partial s}
(
\rme^{(t-s) (\gamma B + C)}
\rme^{s(\gamma b_\ell I+C)}
)
=- \gamma \rme^{(t-s) (\gamma B + C)} (B -  b_\ell I)\rme^{s(\gamma b_\ell I+ C)} ,
\end{equation}
one has
\begin{equation}
\rme^{(t-s) (\gamma B + C)} (B-b_\ell I) = -\frac{1}{\gamma} \left(
\frac{\partial}{\partial s}
(\rme^{(t-s) (\gamma B + C)} \rme^{s(\gamma b_\ell I+ C)})\right)
\rme^{-s(\gamma b_\ell I+ C)},
\label{eq:A10}
\end{equation}
and hence, the integrand~\eqref{eq:A9} is further rewritten as
\begin{equation}
\rme^{(t-s) (\gamma B + C)} (I-P_\ell)C P_\ell \rme^{s (\gamma b_\ell I+ C_Z )} =
-\frac{1}{\gamma}
\left(
\frac{\partial}{\partial s}
(\rme^{(t-s) (\gamma B + C)} \rme^{s(\gamma b_\ell I+ C)})
\right)
\rme^{-s C} S_\ell C P_\ell \rme^{s C_Z }.
\end{equation}
Therefore, the integral in Eq.~\eqref{eq:A3} reads
\begin{align}	
&\rme^{t(\gamma B + C)}P_\ell -\rme^{t (\gamma B+C_Z)} P_\ell
\nonumber\\
&\qquad=-\frac{1}{\gamma}\int_0^t \d s\left(
\frac{\partial}{\partial s}
(\rme^{(t-s) (\gamma B + C)} \rme^{s(\gamma b_\ell I+ C)})
\right)
\rme^{-s C} S_\ell C P_\ell \rme^{s C_Z} 
\nonumber\\
& \qquad = -\frac{1}{\gamma} 
\,\biggl[\rme^{(t-s) (\gamma B + C)} S_\ell C P_\ell \rme^{s(\gamma b_\ell I+ C_Z)}  \biggr]_{s=0}^{s=t}
\nonumber\\
&\qquad\qquad\qquad\qquad\qquad
{}+\frac{1}{\gamma} \int_0^t \d s\, 
\rme^{(t-s) (\gamma B + C)} \rme^{s(\gamma b_\ell I+ C)}  \frac{\d}{\d s} (\rme^{-s C} S_\ell C P_\ell \rme^{s C_Z})
\nonumber\\
& \qquad = \frac{1}{\gamma} 
\,\biggl(
\rme^{t (\gamma B + C)} S_\ell C P_\ell 
- S_\ell C P_\ell \rme^{t(\gamma b_\ell I + C_Z)}  
\nonumber\\
&\qquad\qquad\qquad\qquad\qquad
{}- \int_0^t \d s\, 
\rme^{(t-s) (\gamma B + C)} [C, S_\ell C P_\ell] P_\ell \rme^{s (\gamma b_\ell I+ C_Z)}
\biggr),
\label{eqn:A.13}
\end{align}
where in the second equality we used integration by parts. 

It remains to show that the quantity in the parentheses multiplying $1/\gamma$ in Eq.~(\ref{eqn:A.13}) is uniformly bounded for $t\in[0,t_2]$.
In Kato's proof, this followed simply from unitarity. In our case,
it follows from $b_{\ell}$ being purely imaginary and from $\rme^{t B}$ being a bounded semigroup,
\begin{equation}
\label{eq:boundedetB}
\|\rme^{t   B}\| \le M
\end{equation}
for $t\ge0$ and for some $M\ge1$. This implies that
\begin{equation}
\label{eq:boundedpertsemi}
\|\rme^{t (\gamma B+C)} \| \leq M\rme^{t M\|C\|}
\end{equation}
for $t\ge0$ and $\gamma>0$. Indeed,
\begin{align}
\|\rme^{t (\gamma B+C)} \| 
\le{}&\|\rme^{t \gamma B}\|
+\sum_{n=1}^{+\infty}\int_0^t \d s_1\int_0^{s_1}\d s_2\cdots\int_0^{s_{n-1}}\d s_n\|
\rme^{(t-s_1) \gamma B}C\rme^{ (s_1-s_2) \gamma B}
\cdots
\nonumber\\[-2truemm]
&\qquad\qquad\qquad\qquad\qquad\qquad\qquad\qquad\qquad\qquad\quad
{}\times
\rme^{ (s_{n-1}-s_n) \gamma B}C\rme^{ s_n \gamma B}\|
\nonumber\displaybreak[0]\\
\le{}&
\sum_{n=0}^{+\infty}\frac{t^n}{n!}M^{n+1}\|C\|^n
=M\rme^{t M\|C\|}.
\end{align}
We also have a simpler bound $\|\rme^{tC_Z}\|\le\rme^{t\|C_Z\|}$, and therefore, for $\gamma>0$, Eq.~(\ref{eqn:A.13}) is bounded by
\begin{align}
&\| \rme^{t(\gamma B + C)}P_\ell -\rme^{t (\gamma B+C_Z)} P_\ell \| 
\nonumber\\
&\qquad
\le
\frac{1}{\gamma}
\left(
\|S_\ell C P_\ell\|
\,\Bigl(
M\rme^{tM\|C\|}
+\rme^{t\|C_Z\|}
\Bigr)
 + M\|[C, S_\ell C P_\ell] P_\ell \| 
\frac{\rme^{tM\|C\|}
-\rme^{t\|C_Z\|}
}{M\|C\|-\|C_Z\|}
\right)
\nonumber\\
&\qquad
\le
\frac{1}{\gamma}(M+1)
\|S_\ell C P_\ell\|
\frac{
M
\|C\|
\rme^{tM\|C\|}
-
\|C_Z\| 
\rme^{t\|C_Z\|}
}{M\|C\|-\|C_Z\|}
\nonumber\\
&\qquad
\le
\frac{1}{\gamma}(M+1)
\|S_\ell C P_\ell\|
\frac{
M
\|C\|
\rme^{t_2 M\|C\|}
-
\|C_Z\| 
\rme^{t_2\|C_Z\|}
}{M\|C\|-\|C_Z\|},
\end{align}
for all $t\in[0,t_2]$,
thus proving the uniformity of the limit~\eqref{eq:K29} [for the last two inequalities, we have bounded a factor as $\|[C, S_\ell C P_\ell]P_\ell\|=\|CS_\ell C P_\ell-S_\ell C P_\ell C_Z\|\le\|S_\ell C P_\ell\|(\|C\|+\|C_Z\|)$, 
and used the fact that $t\mapsto \frac{a \rme^{t a}-b\rme^{t b}}{a-b}$ is increasing for positive $t$, for all $a,b>0$].
By summing Eq.~\eqref{eq:K29} over the peripheral spectrum $b_\ell \in \rmi \mathbb{R}$, we get
\begin{equation}
\rme^{t(\gamma B + C)} P_\varphi=\rme^{t \gamma B} \rme^{t C_Z} P_\varphi+\mathcal{O}(1/\gamma)
\quad\text{as}\quad
\gamma\to+\infty,
\label{eq:scl6}
\end{equation}
uniformly for $t\in[0,t_2]$ with $t_2>0$.

\paragraph{Step 2. Eigenvalues with negative real parts:}
Consider now an eigenvalue $b_{\ell}$ with a negative real part, $\Re b_\ell<0$, so that
\begin{equation}
B P_{\ell} = b_{\ell} P_{\ell} + N_{\ell}. \label{eq:eigeqneg}
\end{equation}
We get, for $\gamma>0$ and $t\geq 0$,
\begin{equation}
\rme^{t(\gamma B + C)} = \rme^{t \gamma B} + \int_0^t\d s\, \rme^{(t-s) (\gamma B + C)} C \rme^{s\gamma  B},
\end{equation}
and
\begin{equation}
\label{eqn:NonperipheralDysonBound}
\|\rme^{t(\gamma B + C)} P_\ell - \rme^{t \gamma B} P_\ell\| \le \int_0^t\d s\, \| \rme^{(t-s)(\gamma B + C)}\| \|C\| \|\rme^{s \gamma B} P_\ell \|.
\end{equation}
Now, recalling the bound~(\ref{eq:boundedpertsemi}), we have
$
\|\rme^{(t-s)(\gamma B + C)}\| \le M \rme^{tM \|C\|}
$
for $s\in[0,t]$, and 
\begin{equation}
\rme^{s \gamma B} P_\ell    = \rme^{s\gamma b_\ell} \rme^{s \gamma N_\ell }  P_\ell = \rme^{s\gamma b_\ell} \sum_{k=0}^{n_\ell-1} \frac{1}{k !} (s \gamma N_\ell)^k P_\ell
\end{equation}
with some $n_\ell$, whence
\begin{equation}
\label{eq:exppoly}
\|\rme^{s \gamma B} P_\ell \|  \leq  \rme^{s\gamma \Re b_\ell} p_\ell(s \gamma),
\end{equation}
with $p_\ell$ a polynomial of degree $n_\ell -1$ with positive coefficients depending on $N_\ell$.
Therefore, Eq.~(\ref{eqn:NonperipheralDysonBound}) is bounded by
\begin{align}
\|\rme^{t(\gamma B + C)} P_\ell - \rme^{t \gamma  B} P_\ell\| 
&\le M \|C\| \rme^{tM \|C\|} \int_0^t \d s\,\rme^{s\gamma \Re b_\ell} p_\ell(s \gamma)
\nonumber\\
&\le  \frac{1}{\gamma} M \|C\| \rme^{tM \|C\|} \int_0^{+\infty}\d s\, \rme^{s\Re b_\ell} p_\ell(s).
\end{align}
Thus, we have
\begin{equation}
\rme^{t(\gamma B + C)}P_\ell = \rme^{t \gamma B } P_\ell + \mathcal{O}(1/\gamma)
= \rme^{t \gamma B } \rme^{t  C_Z } P_\ell + \mathcal{O}(1/\gamma)
\end{equation}
as $\gamma \to +\infty$. By summing over the nonperipheral spectrum with $\Re b_\ell <0$, we get
\begin{equation}
\rme^{t(\gamma B + C)}(I-P_\varphi) = \rme^{t \gamma B } \rme^{t  C_Z } (I-P_\varphi)+ \mathcal{O}(1/\gamma)  \quad\text{as}\quad \gamma\to+\infty,
\label{eq:A22}
\end{equation} 
uniformly for $t\in [0, t_2]$ with $t_2>0$.

Now, for $t\geq t_1>0$, we get from Eq.~\eqref{eq:exppoly} that
\begin{equation}
\|\rme^{t \gamma B } P_\ell\| \leq M_1 \rme^{t_1 \gamma  c} \leq \frac{M_1}{t_1 \gamma  |c|}
\end{equation}
with $M_1 \geq 1 $ and $\Re b_\ell \leq c <0$. Therefore, by summing over the nonperipheral spectrum, we have that
\begin{equation}
\rme^{t(\gamma B + C_Z)} (I - P_\varphi) =  \mathcal{O}(1/\gamma) \quad\text{as}\quad
\gamma\to+\infty,
\label{eq:negad}
\end{equation}
uniformly for $t\in [t_1,t_2]$ with  $0<t_1<t_2$.

\paragraph{Step 3. Putting the two cases together:}
Summing Eqs.~(\ref{eq:scl6}) and~(\ref{eq:A22}), we arrive at 
\begin{equation}
\rme^{t(\gamma B+C)}=\rme^{t (\gamma B+C_Z)}+\mathcal{O}(1/\gamma) \quad\text{as}\quad
\gamma\to+\infty,
\label{eq:scl7.0}
\end{equation}
uniformly in $t$ on compact intervals of $[0,+\infty)$,
with $C_Z$ defined in Eq.~(\ref{eq:HZ}).
Since $C_Z$ commutes with $B$, Eq.~(\ref{eq:scl7.0})
implies the strong-coupling limit~(\ref{eq:scl.0}).

Moreover, by using Eq.~(\ref{eq:negad}), which describes how the nonperipheral part decays away at positive times, one finally gets 
\begin{equation}
\rme^{t(\gamma B+C)}=\rme^{t (\gamma B+C_Z)}P_\varphi+\mathcal{O}(1/\gamma) \quad\text{as}\quad
\gamma\to+\infty,
\label{eq:scl7}
\end{equation}
uniformly in $t$ on compact intervals of $(0,+\infty)$. This gives the limit~(\ref{eq:scl}) and concludes the proof.
\end{proof}

\section{Further Estimation of the Error Bound}
\label{app:Bounds}
Here we describe how we get the error bound (\ref{eqn:BoundAdiabaticSimpler}) from the tighter version (\ref{eqn:BoundAdiabatic}).
The idea is to make use of the Jordan normal form to facilitate the estimation of the norms.

We first point out that any matrix can be transformed into the Jordan normal form but with the ``$1$''s just next to the eigenvalues on the diagonal in the Jordan normal form scaled to some constant $\nu$.\footnote{In physics contexts, this rescaling would be somewhat necessary.
For instance, if $t$ in $\rme^{tB}$ is time, $B$ is of the dimension of frequency, and ``$1$'' in the Jordan normal form of $B$ does not make sense without specifying a unit.
The quantity $\nu$ fixes the unit in the Jordan normal form.}
Let us turn $B$ into such a Jordan form by a similarity transformation $T_\nu$,
\begin{equation}
	T_\nu^{-1}BT_\nu=\sum_k(b_k\tilde{P}_k+\tilde{N}_k),
\end{equation} 
where $\{b_k\}$ is the spectrum of $B$, $\{\tilde{P}_k\}$ the diagonal spectral projections, and $\{\tilde{N}_k\}$ the nilpotents with entries $\nu$ or $0$ on the next diagonal.
Notice that the infinity norms (the largest singular values) of $\tilde{P}_k$ and $\tilde{N}_k^n$ are $\|\tilde{P}_k\|_\infty=1$ and $\|\tilde{N}_k^n\|_\infty=\nu^n$ or $0$, respectively, for any positive integer $n$ (in the following, we use infinity norm and omit the subscript ``$\infty$'' of $\|{}\bullet{}\|_\infty$).
Then, we can estimate norms e.g.~as $\|P_k\|\le\|T_\nu^{-1}\|\|\tilde{P}_k\|\|T_\nu\|=\chi_\nu$, where $\chi_\nu=\|T_\nu^{-1}\|\|T_\nu\|\ge1$ is called ``condition number'' \cite{ref:horn}.

We will see that the spectral gaps play important roles, and will realize that it is convenient to choose $\nu$ as
\begin{equation}
	\nu=\min(\eta,\Delta),
\end{equation}
where
\begin{equation}
\eta=\min_{b_k\not\in\rmi\mathbb{R}}|{\Re b_k}|,\qquad
\Delta=\min_{k\neq\ell}|b_k-b_\ell|.
\end{equation}

Now, let us start estimating the norms involved in the error bound (\ref{eqn:BoundAdiabatic}).

\paragraph{(i) Bounding $\bm{\|\rme^{tB}\|\le M}$:}
Recalling the assumption that the peripheral eigenvalues $b_k\in\rmi\mathbb{R}$ of $B$ are semisimple, the spectral representation of $B$ reads
\begin{equation}
\rme^{tB}
=\sum_{b_k\in\rmi\mathbb{R}}\rme^{tb_k}P_k
+\sum_{b_k\not\in\rmi\mathbb{R}}\rme^{tb_k}\rme^{tN_k}P_k.
\end{equation}
It is bounded as
\begin{align}
\|\rme^{tB}\|
&\le
\chi_\nu\sum_{b_k\in\rmi\mathbb{R}}
+\chi_\nu\sum_{b_k\not\in\rmi\mathbb{R}}\rme^{t\Re b_k}\rme^{t\|\tilde{N}_k\|}
\nonumber\\
&\le
\chi_\nu\sum_{b_k\in\rmi\mathbb{R}}
+\chi_\nu\sum_{b_k\not\in\rmi\mathbb{R}}\rme^{-\eta t}\rme^{\nu t}
\nonumber\\
&\le
\chi_\nu\sum_{b_k\in\rmi\mathbb{R}}
+\chi_\nu\sum_{b_k\not\in\rmi\mathbb{R}}
\nonumber\\
&
\le D\chi_\nu
\vphantom{\sum_{b_k\not\in\rmi\mathbb{R}}}
\nonumber\\
&
\equiv M,
\end{align}
where we have used $\nu\le\eta$ and $D$ is the dimension of $B$.

\paragraph{(ii) Bounding the decaying part $\bm{\|\rme^{tB}(I-P_\varphi)\|\le\rme^{-\eta t}p(t)}$:}
The spectral representation of $B$ yields
\begin{equation}
\rme^{tB}(I-P_\varphi)
=\sum_{b_k\not\in\rmi\mathbb{R}}\rme^{tb_k}\rme^{tN_k}P_k
=\sum_{b_k\not\in\rmi\mathbb{R}}\rme^{tb_k}
\sum_{n=0}^{n_k-1}\frac{1}{n!}t^nN_k^nP_k.
\end{equation}
It is bounded as
\begin{equation}
\|\rme^{tB}(I-P_\varphi)\|
\le\chi_\nu
\rme^{-\eta t}
\sum_{b_k\not\in\rmi\mathbb{R}}
\sum_{n=0}^{n_k-1}\frac{1}{n!}(\nu t)^n
\equiv\rme^{-\eta t}p(t).
\end{equation}
We need its integral:
\begin{align}
\int_0^\infty\d s\,\rme^{-\eta s}p(s)
&\le
\chi_\nu
\frac{1}{\eta}
\sum_{b_k\not\in\rmi\mathbb{R}}
\sum_{n=0}^{n_k-1}
(\nu/\eta)^n
\nonumber\\
&
\le
\chi_\nu
\frac{1}{\eta}
\sum_{b_k\not\in\rmi\mathbb{R}}
n_k
\nonumber\\
&
\le
\chi_\nu
\frac{D}{\eta}
\nonumber\\
&
=
\frac{M}{\eta},
\end{align}
where we have used $\nu\le\eta$ and $n_k\le \rank P_k$.

\paragraph{(iii) Bounding the reduced resolvent $\bm{\|S_\ell\|}$:}
The reduced resolvent $S_\ell$ defined in Eq.~(\ref{eqn:Resolvent}) is cast into
\begin{equation}
S_\ell
=\sum_{k\neq\ell}[(b_k-b_\ell)I+N_k]^{-1}P_k
=\sum_{k\neq\ell}
\sum_{n=0}^{n_k-1}
\frac{(-1)^n}{(b_k-b_\ell)^{n+1}}N_k^n
P_k.
\end{equation}
It is bounded as
\begin{align}
\|S_\ell\|
&\le\chi_\nu
\sum_{k\neq\ell}
\sum_{n=0}^{n_k-1}
\frac{\nu^n}{|b_k-b_\ell|^{n+1}}
\nonumber\\
&\le\chi_\nu
\frac{1}{\Delta}
\sum_{k\neq\ell}
\sum_{n=0}^{n_k-1}
(\nu/\Delta)^n
\nonumber\\
&
\le\chi_\nu
\frac{1}{\Delta}
\sum_{k\neq\ell}
n_k
\nonumber\\
&
\le\chi_\nu
\frac{D}{\Delta}
\nonumber\\
&
=\frac{M}{\Delta},
\end{align}
where we have used $\nu\le\Delta$.

\paragraph{(iv) Bounding the projected generator $\bm{\|C_Z\|}$:}
Let us also try to simplify the factor involving $\|C_Z\|$ in the bound (\ref{eqn:BoundAdiabatic}).
The projected generator $C_Z$ is defined in Eq.~(\ref{eq:HZ}), and is bounded as
\begin{equation}
\|C_Z\|
\le\chi_\nu^2\sum_{b_k\in\rmi\mathbb{R}}\|C\|
\le D\chi_\nu^2\|C\|
=\chi_\nu M\|C\|.
\end{equation}
Now, by noting the inequality $(\sinh x)/x\le\cosh x$, we can bound as
\begin{align}
\frac{a\rme^{at}-b\rme^{bt}}{a-b}	
&=
\left(
\cosh[(a-b)t/2]+(a+b)\frac{\sinh[(a-b)t/2]}{a-b}
\right)
\rme^{(a+b)t/2}
\nonumber\\
&\le
\left(
1+\frac{1}{2}t(a+b)
\right)
\rme^{(a+b)t/2}
\cosh[(a-b)t/2]
\nonumber\\
&=\frac{1}{2}
\left(
1+\frac{1}{2}t(a+b)
\right)
(
\rme^{at}+\rme^{bt}
)
\end{align} 
for $a,b\ge0$ and $t\ge0$.
By using this inequality, we can get
\begin{align}
\frac{M\|C\|\rme^{tM\|C\|}-\|C_Z\|\rme^{t\|C_Z\|}}{M\|C\|-\|C_Z\|}	
&\le\frac{1}{2}
\left[
1+\frac{1}{2}t\,\Bigl(M\|C\| + \|C_Z\|\Bigr)
\right]
(
\rme^{t M\|C\|}+\rme^{t \|C_Z\|}
)
\nonumber\\
&\le\frac{1}{2}
\left[
1+\frac{1}{2}t\,\Bigl(M\|C\| + \chi_\nu M\|C\|\Bigr)
\right]
(
\rme^{t M\|C\|}+\rme^{t \chi_\nu M\|C\|}
)
\nonumber\\
&\le
\Bigl(
1+t\chi_\nu M\|C\|
\Bigr)\,
\rme^{t \chi_\nu M\|C\|}
\nonumber\\
&\le
\rme^{2t \chi_\nu M\|C\|}.
\end{align}

\bigskip
Gathering all these estimates, we can bound Eq.~(\ref{eqn:BoundAdiabatic}) as
\begin{align}
&\|\rme^{t(\gamma B + C)}-\rme^{t\gamma B}\rme^{tC_Z}P_\varphi\|
\nonumber\\
&\qquad
\le
\frac{1}{\gamma}\left(
(M+1)\sum_{b_\ell\in\rmi\mathbb{R}}\chi_\nu\frac{M}{\Delta}\|C\|
\rme^{2t \chi_\nu M\|C\|}
+\frac{M^2}{\eta}\|C\|\rme^{tM\|C\|}
\right)
\nonumber\\
&\qquad\qquad\qquad\qquad\qquad\qquad\qquad\qquad\qquad\quad
{}+\chi_\nu
\rme^{-\gamma\eta t}
\sum_{b_k\not\in\rmi\mathbb{R}}
\sum_{n=0}^{n_k-1}\frac{1}{n!}(\gamma\nu t)^n
\nonumber\\
&\qquad
\le
\frac{1}{\gamma}M^2\|C\|
\left(
\frac{M+1}{\Delta}
\rme^{2t \chi_\nu M\|C\|}
+\frac{1}{\eta}
\rme^{tM\|C\|}
\right)
+M
\rme^{-\gamma\eta t}
\sum_{n=0}^{D-1}\frac{1}{n!}(\gamma\nu t)^n
\nonumber\\
&\qquad
\le
\frac{1}{\gamma}M^2\|C\|
\left(
\frac{2M}{\Delta}
+\frac{1}{\eta}
\right)
\rme^{2tM^2\|C\|}
+M
\rme^{-\gamma\eta t}
\sum_{n=0}^{D-1}\frac{1}{n!}(\gamma\eta t)^n.
\end{align}
This is the bound presented in Eq.~(\ref{eqn:BoundAdiabaticSimpler}).


\end{document}